\documentclass[aoas, preprint]{imsart}

\usepackage{amsmath}
\usepackage{amssymb}
\usepackage{amsthm}
\usepackage{booktabs}
\usepackage{caption}
\usepackage[T1]{fontenc}
\usepackage{enumitem}
\newlist{inlinelist}{enumerate*}{1}
\setlist*[inlinelist]{label=(\arabic*)}
\usepackage[totalwidth=450pt, totalheight=630pt]{geometry}
\usepackage{graphicx}
\PassOptionsToPackage{hyphens}{url} % pass options to url autoloaded by hyperref
\usepackage[pdfencoding=auto, psdextra]{hyperref} % hyperref must be loaded before cleveref
\usepackage{mathrsfs}
\usepackage{multirow}
\usepackage[round]{natbib} % natbib must be loaded before algorithm2e
\usepackage[linesnumbered, ruled]{algorithm2e} % algorithm2e must be loaded before cleveref
\usepackage{cleveref}
\usepackage{subcaption}
\usepackage{tabularx}
\usepackage{tikz}
\usetikzlibrary{calc}
\usepackage{verbatim}
\usepackage{xcolor}

\arxiv{arXiv:1903.08747}

\startlocaldefs

\definecolor{deep-blue}{HTML}{2c7bb6}
\colorlet{mid-blue}{deep-blue!25}

\newtheorem{corollary}{Corollary}
\newtheorem{definition}{Definition}
\newtheorem{lemma}[corollary]{Lemma}
\newtheorem{proposition}[corollary]{Proposition}

\theoremstyle{definition}
\newtheorem{assumption}{Assumption}
\newtheorem{example}{Example}
\newtheorem*{remark}{Remark}

\newtheoremstyle{custom}{}{}{\itshape}{}{\bfseries}{.}{.5em}{\thmnote{#3}}
\theoremstyle{custom}

\setlist{listparindent=\parindent, parsep=0pt}

\newcommand{\EE}{\mathbb{E}}

\newcommand{\PP}{\mathbb{P}}
\newcommand{\RR}{\mathbb{R}}

\DeclareMathOperator{\sgn}{sign}
\DeclareMathOperator{\Binomial}{Binomial}
\DeclareMathOperator{\Uniform}{Uniform}

\newcommand{\htheta}{\hat{\theta}}
\newcommand{\FDPd}{\text{FDP}_\text{dir}}
\newcommand{\hFDPd}{\widehat{\text{FDP}}_\text{dir}}
\newcommand{\ao}{\alpha_0}
\newcommand{\lest}{\le_\text{st}}
\newcommand{\gest}{\ge_\text{st}}

\endlocaldefs

\graphicspath{{./fig/}}

\begin{document}

\begin{frontmatter}

\title{Statistical Methods for Replicability Assessment}
\runtitle{Statistical Methods for Replicability Assessment}

\begin{aug}
  \author{\fnms{Kenneth} \snm{Hung}\corref{}\ead[label=e1]{kenhung@berkeley.edu}}
  \and
  \author{\fnms{William} \snm{Fithian}\ead[label=e2]{wfithian@berkeley.edu}}

  \runauthor{K. Hung and W. Fithian}

  \affiliation{University of California, Berkeley}

  \address{Department of Mathematics\\ 951 Evans Hall, Suite 3840\\ Berkeley, CA 94720-3840\\
  \printead{e1}}

  \address{Department of Statistics\\ 301 Evans Hall\\ Berkeley, CA 94720\\
          \printead{e2}}

\end{aug}

\begin{abstract}
  Large-scale replication studies like the Reproducibility Project: Psychology (RP:P) provide invaluable systematic data on scientific replicability, but most analyses and interpretations of the data fail to agree on the definition of ``replicability'' and disentangle the inexorable consequences of known selection bias from competing explanations. We discuss three concrete definitions of replicability based on
  \begin{inlinelist}
    \item whether published findings about the signs of effects are mostly correct,
    \item how effective replication studies are in reproducing whatever true effect size was present in the original experiment, and
    \item whether true effect sizes tend to diminish in replication.
  \end{inlinelist}
  We apply techniques from multiple testing and post-selection inference to develop new methods that answer these questions while explicitly accounting for selection bias. Our analyses suggest that the RP:P dataset is largely consistent with publication bias due to selection of significant effects. The methods in this paper make no distributional assumptions about the true effect sizes.
\end{abstract}



\end{frontmatter}

\section{Introduction}
\label{sec:intro}

  Growing concerns about selection bias, $p$-hacking, and other questionable research practices (QRPs) have raised urgent questions about the reliability of scientific findings. While concerns about replicability cut across scientific disciplines, psychologists have led large-scale efforts to assess the replicability of their own field. The largest and most systematic of these efforts has been the Reproducibility Project:\ Psychology (RP:P),\footnote{In some parts of the literature, ``reproducibility'' has taken on a computational connotation, meaning only that other scientists can repeat the analysis using the original study's data; we will lean toward the more unambiguous term ``replicability.''} a major collaboration by several hundred psychologists to replicate a representative sample of 100 studies published in 2008 in three top psychology journals, {\em Psychological Science}, {\em Journal of Personality and Social Psychology}, and {\em Journal of Experimental Psychology: Learning, Memory, and Cognition}.\footnote{The test statistics, effect sizes and most pertinent information are all publicly available on at the Open Science Foundation website at \url{https://osf.io/ezcuj/}.}

  While the RP:P dataset is an invaluable resource, scientists disagree on how to quantify or measure replicability \citep{Goodman:2016bo,Amrhein:2017en}. \citet[OSC;][]{OpenScienceCollaboration:2015cn} reported three main metrics: it found that $64\%$ ($= 1 - 36\%$) of the replication studies did not find statistically significant results in the same direction as the original studies, that $53\%$ ($= 1 - 47\%$) of $95\%$ confidence intervals for the replication studies do not contain the point estimates for their corresponding original studies, and that $83\%$ of the effect size estimates declined from original studies to replications. All three summary statistics were widely reported as indicating a dire crisis for the credibility of experimental psychology research. For example, the {\em Washington Post} reported that RP:P ``affirms that the skepticism [of published results] was warranted'' \citep{Achenbach:2015vi}; the {\em Economist} noted that OSC ``managed to replicate satisfactorily the results of only $39\%$ of the studies investigated'' \citeyearpar{Anonymous:uoQKTVTm}; and the {\em New York Times} reported that ``more than half of the findings did not hold up when retested'' \citep{Carey:2015wp}.

  This negative gloss was challenged in a comment by \citet{Gilbert:2016he}, who criticized both the fidelity of some of the replications' experimental designs and the aptness of the metrics reported by \citet{OpenScienceCollaboration:2015cn}. In particular, \citeauthor{Gilbert:2016he}\ pointed out that, because there is sampling error in the replication point estimates, we should not expect $95\%$ of the estimates to fall into the replication confidence intervals even under ideal conditions. Moreover, any small or large variations in the true effect sizes between the original and replication studies could further deflate the expected fraction of ``successful replications,'' as measured in this way. \citeauthor{Gilbert:2016he}\ concluded that ``OSC seriously underestimated the reproducibility of psychological science,'' sparking further debate between defenders of OSC's conclusions \citep{Anderson:2016gs,Srivastava:2016,Nosek:2016} and the critics \citep{Gilbert:2016uv,Gilbert:2016th}.\footnote{While much of the ensuing discussion focused on the question of whether the confidence interval metric $53\%$ is too pessimistic, analogous criticisms apply to the ``significant replications'' metric of $64\%$ as well: the replication studies could be underpowered even when a true effect is present.}

\subsection{Three definitions of replicability}

  To determine whether OSC truly underestimated replicability, we must first pin down the rather slippery question of what ``replicability'' actually is. Although the three metrics used by OSC are simply descriptive statistics that do not purport to estimate any explicitly defined underlying quantity, we can loosely characterize the $64\%$, $53\%$ and $83\%$ numbers respectively as qualitative answers to three questions:
  \begin{description}
    \item[False directional claims.] {\em What fraction of the original studies were erroneous} in claiming that the true effect was nonzero, in the claimed direction (positive or negative)? \citet{Gelman:2000tg} called such mistakes {\em type S} errors.
    \item[Effect shift.] {\em How much do the effect sizes shift} from the original study to the replication study? We call the discrepancy between the original and replication effect {\em effect shift}.
    \item[Effect decline.] {\em What fraction of the effect sizes decline?} More precisely, what fraction of the true effect sizes shift in a direction opposite to the original claims when the studies were replicated, and by how much?
  \end{description}
  The first question concerns a type of {\em false discovery rate} (FDR) of the statistical hypotheses, viewing the field of social psychology as a collective enterprise in large-scale multiple testing: it quantifies the fraction of findings that would be confirmed if the exact same studies could be carried out again with much larger samples from the same populations. The second question concerns a basic form of repeatability: whether scientists are typically successful in closely replicating each others' experimental conditions, so that the true effect being measured is stable across different experiments. The third question builds upon the second question: whether true effect sizes tend systematically to attenuate in replications. An overall trend of declining true effects could suggest various interpretations, including systematic biases in the original experiments or failures by the replication teams to reproduce key experimental conditions that produced the original effects.

  As we will see, however, none of the three reported metrics can be taken at face value as {\em estimates} of the answers to the corresponding questions, due to the confounding factor of pervasive selection bias. By using techniques from multiple testing and post-selection inference, we will develop methods to rigorously address these questions without assuming a model for the prior distribution of effect sizes. For the RP:P data we estimate the rate of false directional claims at roughly $32\%$ among studies with $p < 0.05$, which would be considered unacceptably high in most multiple testing applications. By contrast, among studies with $p < 0.005$, a lower threshold proposed by \citet{Benjamin:2018gh}, our estimate drops to $7\%$, with an upper confidence bound of $18\%$. We also compute confidence intervals for the effect shift in each individual study pair and find that, after adjusting for multiplicity, about $11\%$ of the intervals exclude zero, an idealized null hypothesis of perfect replication. For effect decline, we find in aggregate that $35\%$ of the true effects declined, and $22\%$ declined by at least $25\%$.

  In addressing each question, we define our estimands in terms of the true effects present in the statistical populations actually sampled in each study. Because some studies may be biased or lack external validity --- for example, because of flaws in the study design, or because survey participants are unrepresentative of the broader population of scientific interest --- these effect sizes may not reflect the latent scientific quantities the experiments purport to measure. Uncovering such discrepancies is beyond the reach of data analysis alone, but we should keep them in mind as we interpret the results.

\subsection{The role of selection bias}

  The RP:P data shows unmistakable signs of selection for statistically significant findings in the original experiments: $91$ of the $100$ results replicated by OSC were statistically significant at the $0.05$ level in the original study and four of the others had ``marginally significant'' $p$-values between $0.05$ and $0.06$. This is due partly to publication bias (that the studies might not have been published, or the results discussed, if the $p$-values had not been significant), but also partly to OSC's method for choosing which results to replicate. Each OSC replication team selected a ``key result'' from the last experiment presented in the original paper, and evidently most teams chose a significant finding as the key result (justifiably so, since positive results usually draw the most attention from journal readers and the outside world). \Cref{fig:pval-dist} shows the empirical distribution of $p$-values from the original and replication studies.
  \begin{figure}[htbp]
    \centering
    \includegraphics[width=0.8\textwidth]{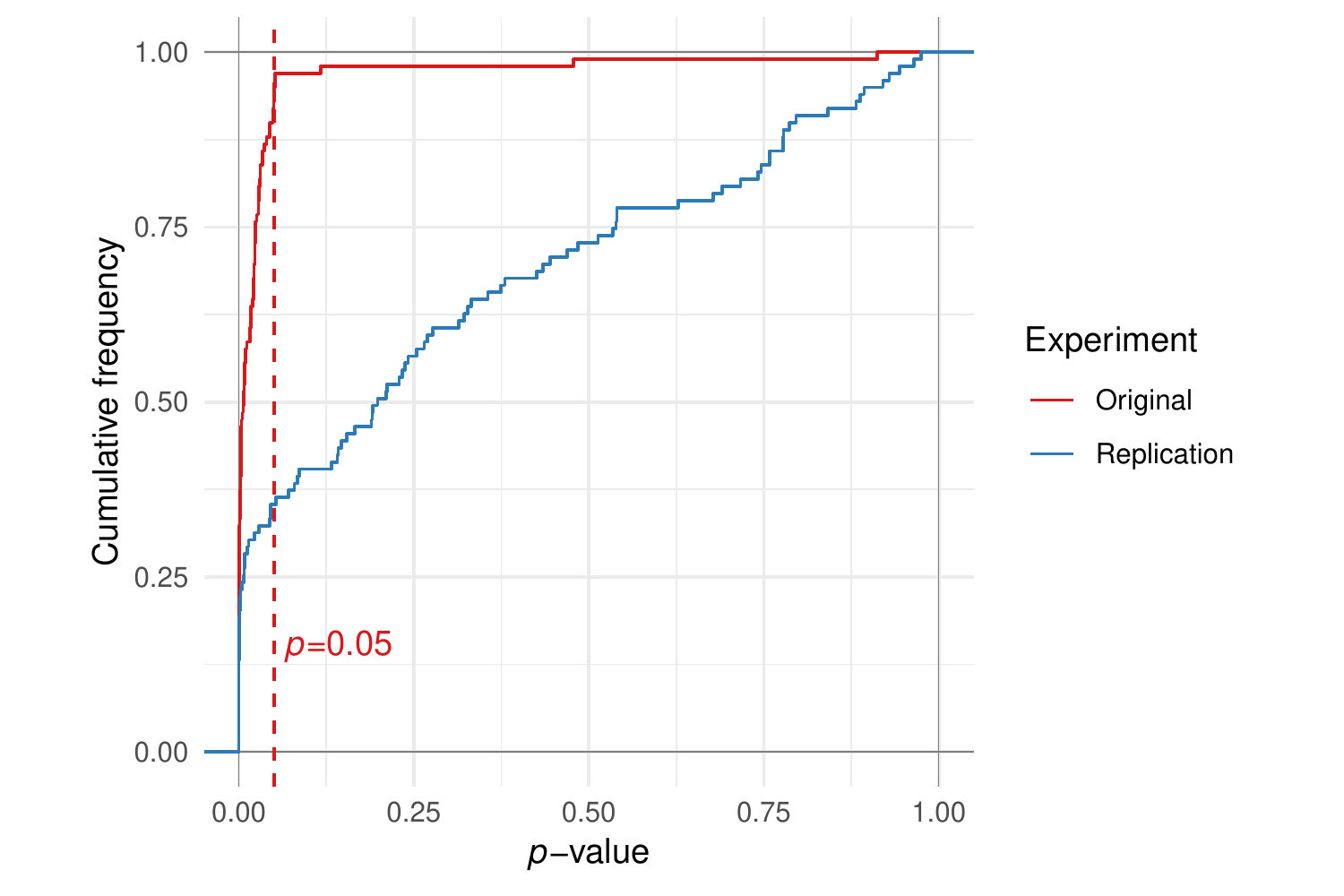}
    \caption{The empirical distribution of the original and replication $p$-values. Nearly all of the original $p$-values (in red) are smaller than $0.05$.}
  \label{fig:pval-dist}
  \end{figure}

  The resulting selection bias in the original studies leads to many well-known and predictable pathologies, such as systematically inflated effect size estimates, undercoverage of (unadjusted) confidence intervals, and misleading answers from unadjusted meta-analyses. Indeed, most of the phenomena reported by OSC, including the three metrics discussed above, could easily be produced by selection bias alone. This would be true {\em even if there are few false directional claims, all replications are exact, and true effects do not decline}, as illustrated in the following simulation study.
  
  \begin{example}
    \label{eg:sel-bias}
    Consider a stylized setting where all experiments (both original and replication) have an identical effect size $\theta$, producing an unbiased Gaussian estimate with standard error 1. Assume, however, that we observe only study pairs for which the original study is significant at level $0.05$.

    \Cref{fig:naive-same-dir-func} shows the expected fraction of replication studies which are not statistically significant in the same direction as the corresponding original studies, as a function of effect size $\theta$, along with the true proportion of false directional claims; or type S errors. Even when the true error rate is low, e.g.\ at $\theta = 1$ as shown in \Cref{fig:naive-same-dir-theta1}, the proportion of replications reporting the same directional findings as the original studies can remain low.
    \begin{figure}[htbp]
	    \centering
	    \begin{subfigure}[t]{0.59\hsize}
	      \centering
	      \includegraphics[width=\hsize]{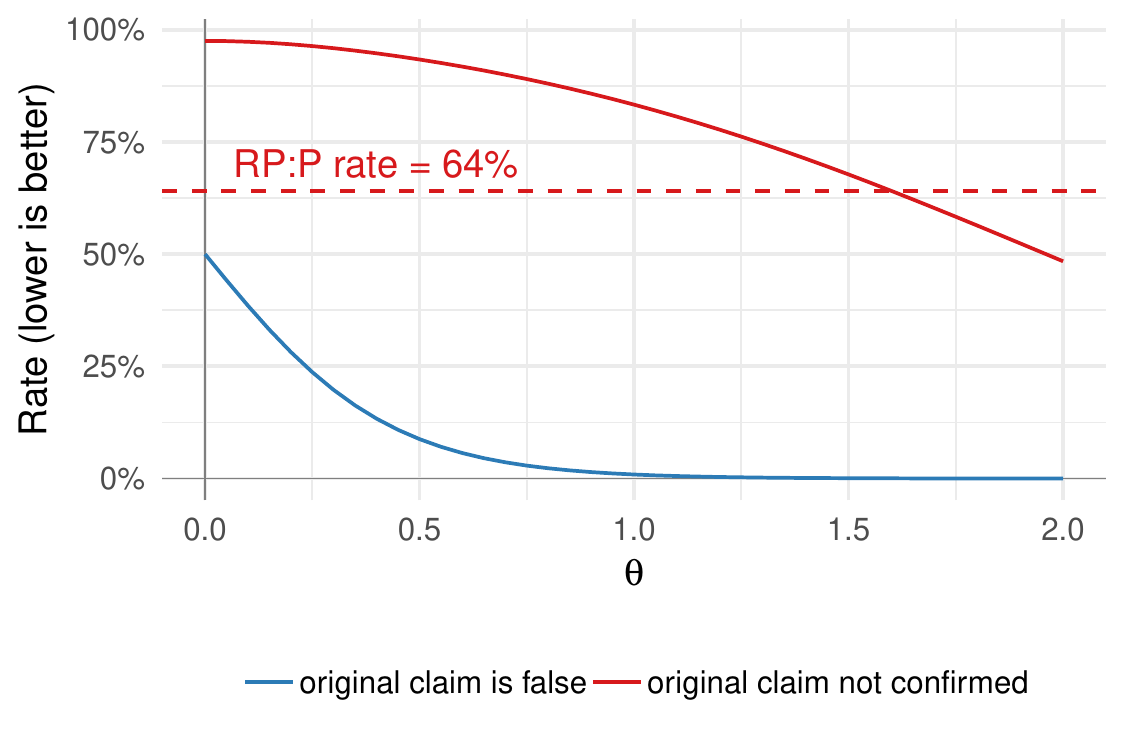}
	      \caption{The expected fraction of replications that do not confirm (at level $0.05$) the original directional claim (red), and the proportion of false directional claims in the original studies (blue), as a function of effect size $\theta$. For small $\theta$, the fraction of replications that do not confirm the claims in the original studies may dramatically overestimate the fraction of false original claims.}
	    \label{fig:naive-same-dir-func}
	    \end{subfigure}
	    \hfill
	    \begin{subfigure}[t]{0.39\hsize}
	      \centering
	      \includegraphics[width=\hsize]{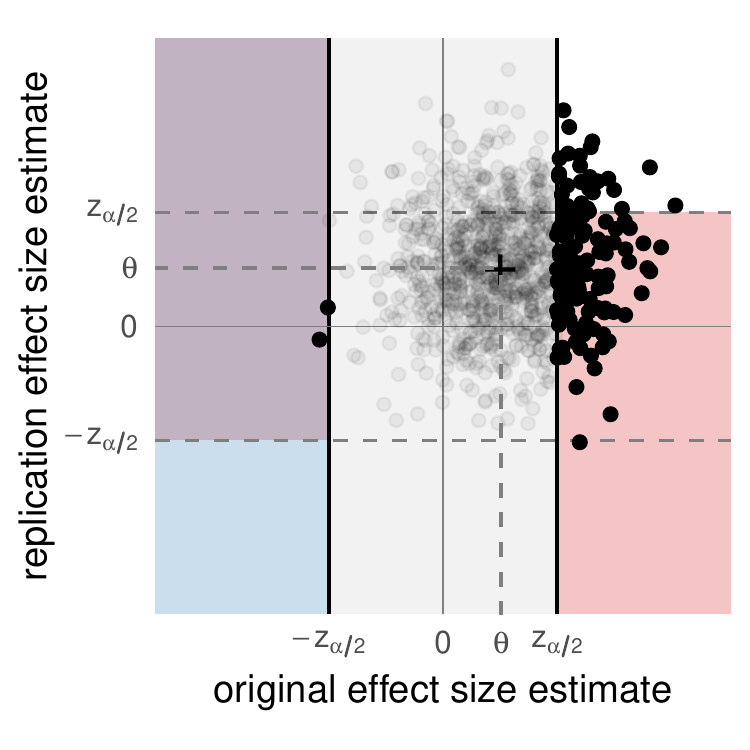}
	      \caption{$\theta=1$. The gray region is unobserved. For points in the red region, the replication does not confirm the original directional claim, and for points in the blue region, the original claim is directionally false. The red and blue regions overlap in the purple region.}
	    \label{fig:naive-same-dir-theta1}
	    \end{subfigure}
	    \caption{}
	  \label{fig:naive-same-dir}
	  \end{figure}

    Likewise, we simulate the expected fraction of $95\%$ replication confidence intervals that fail to cover their original point estimates in \Cref{fig:naive-ci-sim} and the expected fraction of effect sizes that decline in \Cref{fig:naive-decline}. In both cases, we see that selection bias is more than sufficient to produce the metrics in RP:P, even in our idealized simulation with exact replications and relatively few type S errors.
    \begin{figure}[htbp]
	    \centering
	    \begin{subfigure}[t]{0.59\hsize}
	      \centering
	      \includegraphics[width=\hsize]{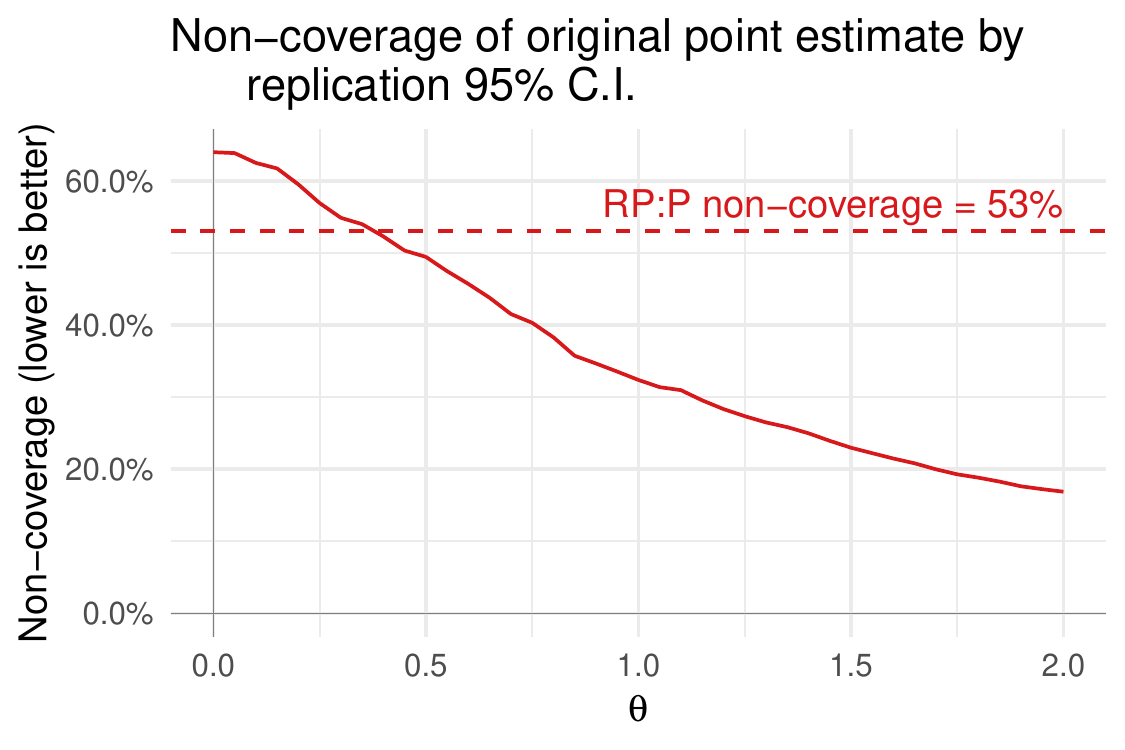}
	      \caption{Expected fraction of original point estimates falling outside the replication confidence interval, as a function of effect size $\theta$. For small $\theta$, the fraction of original point estimates falling outside the replication $95\%$ confidence intervals can easily exceed the RP:P reported metric of $53\%$, even when all replications are perfectly exact.}
	    \label{fig:naive-ci-sim-func}
	    \end{subfigure}
	    \hfill
	    \begin{subfigure}[t]{0.39\hsize}
	      \centering
	      \includegraphics[width=\hsize]{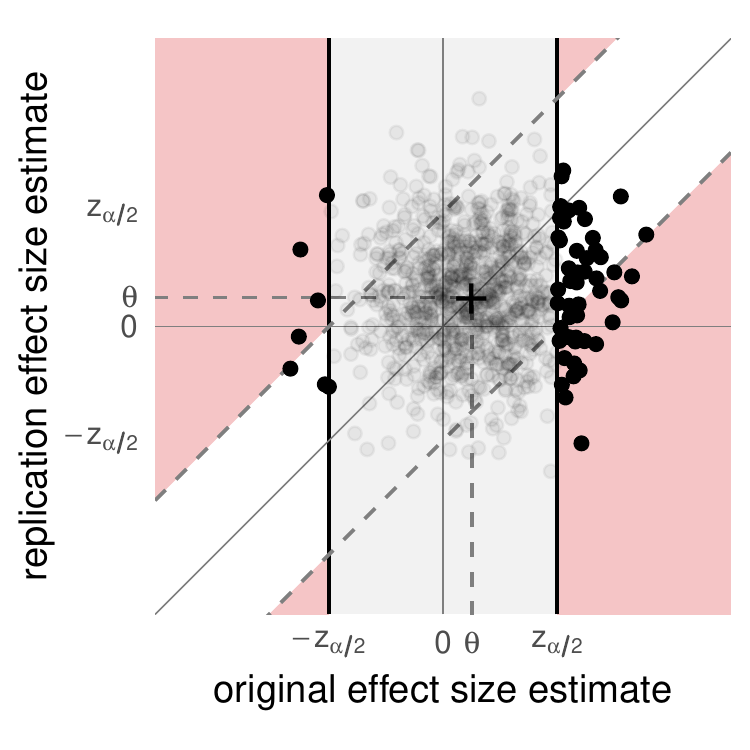}
	      \caption{$\theta = 0.5$. The gray region is unobserved. For points in the red region, the original point estimate differs from the replication estimate by more than $z_{\alpha/2}$ and hence the original point estimate falls outside in the replication $95\%$ confidence interval.}
	    \label{fig:naive-ci-theta05}
	    \end{subfigure}
	    \caption{}
	  \label{fig:naive-ci-sim}
	  \end{figure}
	  \begin{figure}[htbp]
	    \centering
	    \begin{subfigure}[t]{0.59\hsize}
	      \centering
	      \includegraphics[width=\hsize]{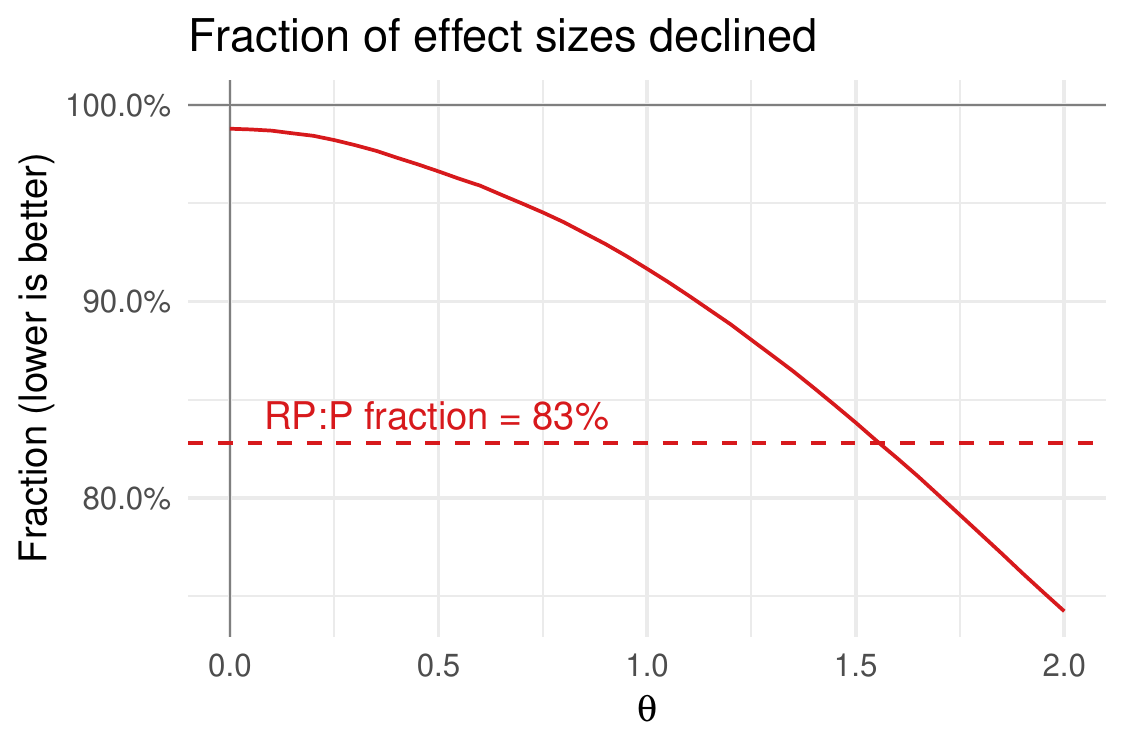}
	      \caption{Expected fraction of effect size point estimates that declined toward zero in replication, as a function of effect size of $\theta$. For small $\theta$, the fraction of effect size estimates declining from original to replication studies can easily exceed the RP:P reported metric of $83\%$, even when there is no decline in the true effect sizes.}
	    \label{fig:naive-decline-func}
	    \end{subfigure}
	    \hfill
	    \begin{subfigure}[t]{0.39\hsize}
	      \centering
	      \includegraphics[width=\hsize]{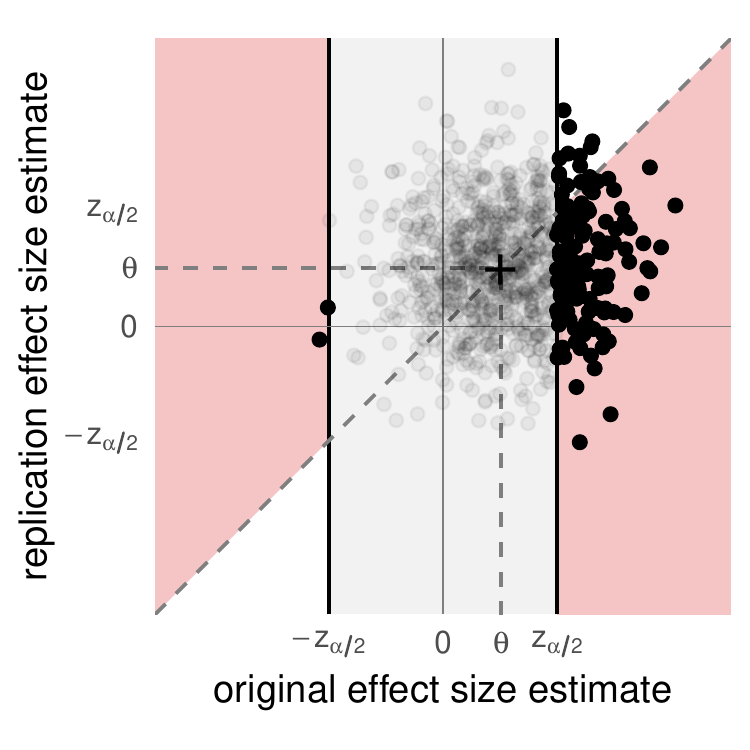}
	      \caption{$\theta = 0.5$. The gray region is unobserved. Points in the red region represent declining point estimates in replications. When the original point estimate is positive, a decline is marked by a smaller replication estimate; on the other hand, if the original estimate is negative, a decline is indicated by a larger replication estimate.}
	    \label{fig:naive-ci-theta03}
	    \end{subfigure}
	    \caption{}
	  \label{fig:naive-decline}
	  \end{figure}
  \end{example}

  Because selection bias could, in principle, provide a sufficient explanation for the metrics reported in RP:P, those metrics do not, in and of themselves, provide any evidence of any other problems. In particular, they shed no light on whether the FDR is actually high, or how much the effect sizes shifted, or whether effect sizes tend to decline. Nor do they provide evidence for any competing accounts of the replication crisis, such as QRPs like $p$-hacking, high between-study variability in effect sizes, or systematic biases in the original studies. To discern anything about other explanations, we must adjust for the pervasive effects of selection bias.

  Another good reason to disentangle selection bias from other sources of error is that the former is, in some sense, the most innocuous explanation for the phenomena observed by OSC while the others present much deeper scientific issues. The technical issues of selection bias can be addressed either retrospectively by statistical adjustments \citep[e.g.][]{Duval:2000dg,Hedges:1992eb,Simonsohn:2014ch,Fithian:2014ws,Andrews:2018vh}, or prospectively with more preregistration or larger sample sizes. By contrast, it would be deeply worrying if psychologists were systematically unable to repeat their colleagues' experiments, or if most published claims about effect sizes were directionally incorrect.
  
\subsection{Formalizing replicability}

  We now introduce a simple formal model for replication studies with selection bias. For study $i=1, \ldots, m$, let $\theta_{i,O}$ and $\theta_{i,R}$ denote the true effect sizes in the original and the replication studies, respectively. Abstracting away experimental design details, assume that each study pair produces two normally distributed effect size estimators $\htheta_{i,O}$ and $\htheta_{i,R}$. Assume additionally that for the study pair to appear in our replication data, $\htheta_{i,O}$ must be statistically significant at level $\alpha=0.05$;\footnote{We relax this assumption in \Cref{sec:methodology}.} then for some significance threshold $c>0$ we have
  \begin{equation}
    \htheta_{i,O} \sim N\left(\theta_{i,O}, \sigma_{i,O}^2\right) 1_{\{|\htheta_{i,O}| > c\}} \qquad\text{ and }\qquad \htheta_{i,R} \sim N\left(\theta_{i,R}, \sigma_{i,R}^2\right),
  \label{eq:trunc-model}
  \end{equation}
  with all estimates assumed to be independent of each other. The indicator $1_{\{|\htheta_{i,O}| > c\}}$ beside the normal distribution in \eqref{eq:trunc-model} means that the distribution of $\htheta_{i,O}$ has been truncated to the event where $|\htheta_{i,O}|>c$ and renormalized so that it integrates to $1$. For the moment, we assume that the variances $\sigma_{i,O}^2$ and $\sigma_{i,R}^2$ are known; in that case $c = z_{0.05 / 2}\, \sigma_{i,O}$. We will relax this assumption in \Cref{sec:methodology}.

  \paragraph{False directional claims} To formalize false directional claims in terms of the parameters of model~\eqref{eq:trunc-model}, we note that a type S error occurs when a statistically significant finding gets the sign of the parameter wrong:
  \[
    H_i^{S,O}:\; \sgn(\theta_{i,O}) \ne \sgn(\htheta_{i,O}), \qquad\text{ where } \sgn(x) =
    \begin{cases}
      +1, & x > 0 \\
      -1, & x < 0 \\
      0, & x = 0
    \end{cases}.
  \]
  Note that $|\htheta_{i,O}|$ is always larger than $c$, so $\sgn(\htheta_{i,O}) \in \{-1,+1\}$. Letting $S_i = \sgn(\htheta_{i,O})$, we can rewrite the hypothesis as
  \[
    H_i^{S,O}:\; S_i\,\cdot\, \theta_{i,O} \le 0.
  \]
  Here $H_i^{S,O}$ is fundamentally data-dependent as it is determined by $S_i$. Nonetheless it is a meaningful hypothesis: when $S_i = +1$, we want to test the null that $\theta_{i,O} \le 0$; otherwise we want to test the null that $\theta_{i,O} \ge 0$. Our strategy is to condition on the value of $S_i$, since the null hypothesis is fixed again once we know $S_i$. We defer the discussion of valid testing of data-dependent hypotheses for now.

  The question of false directional claims, then, boils down to asking how many $H_i^{S,O}$ are true: a multiple testing problem. Our estimand, the proportion of type S errors that occurred, is $V / R$, where $V$ is the number of type S errors and $R$ is the number of ``discoveries,'' i.e.\ rejections. If we classify the hypotheses by whether $H_i^{S,O}$ is true and whether the test for $H_i^{S,O}$ is significant, then $V$ and $R$ correspond to the cell counts in \Cref{tbl:err}.
  \begin{table}[htbp]
    \centering
    \begin{tabular}{lccc}
      \toprule
      Original $p$-value & $H_i^{S,O}$ is true & $H_i^{S,O}$ is false & Total \\
      \midrule
      Significant & $V$ & $*$ & $R$ \\
      Not-significant & $*$ & $*$ & $*$ \\
      Total & $*$ & $*$ & $*$ \\
      \bottomrule
    \end{tabular}
    \caption{Classification of the hypotheses, in the style of \citet{Benjamini:1995cd}. Only $R$ is observed and we wish to infer on $V$.}
  \label{tbl:err}
  \end{table}

  In the multiple testing literature, $V/R$ is called the {\em directional false discovery proportion} (directional FDP, or $\FDPd$), the type S error analog of false discovery proportion \citep[FDP;][]{Benjamini:2000ka}. In addition to an estimate, we also provide an upper confidence bound for the directional FDP in \Cref{sec:methodology}. Both the estimator and the confidence bound are based on a ``$p$-curve'' analysis, i.e.\ an analysis of the distribution of significant $p$-values \citep{Simonsohn:2014fa}. We further modify these methods to evaluate the proposal to lower the statistical significance threshold by \citet{Benjamin:2018gh}.

  Although $\htheta_{i,R}$ is irrelevant to testing $H_i^{S,O}$, it is informative for the closely related question of whether $\htheta_{i,O}$ incorrectly predicts the direction of the effect in a replication study, i.e.
  \[
    H_i^{S,R}:\; S_i \,\cdot\, \theta_{i,R} \le 0.
  \]
  Note that $S_i$ is computed from the original study, so this hypothesis is a measure of external validity as to the (claimed) directions of effects. If an experimental result has external validity, then any directional claim about the true effect should apply not only to the original study, but also to direct replications thereof. We provide analogous methods for multiple testing of the hypotheses $H_i^{S,R}$.

  \paragraph{Effect shift} To assess the effect shift in a specific replication attempt, we can test the hypothesis $H_i^E:\; \theta_{i,O} = \theta_{i,R}$ (an exact replication). As \citet{Anderson:2016gs} noted, ``there is no such thing as exact replication''; nevertheless, exactness serves usefully as an idealized null hypothesis. By inverting a test for $H_i^E$ we can obtain a predictive interval for $\htheta_{i,R}$. Furthermore, by inverting tests for a related hypothesis $H_i^{E,\delta}:\; \theta_{i,O} - \theta_{i,R} = \delta$, we obtain a confidence interval for $\theta_{i,O} - \theta_{i,R}$, the effect shift in study $i$. Our methods explicitly take into account the truncation of $\htheta_{i,O}$.

  \paragraph{Effect size decline} The null hypothesis for effect size decline is closely related to effect shift, and can be formalized as the null hypothesis where the true effect size has declined by no more than a fraction $\rho \in [0, 1]$:
  \[
    H_i^{D,\rho}:\; S_i \,\cdot\, \theta_{i,R} \ge S_i \,\cdot\, (1 - \rho) \theta_{i,O}.
  \]
  If $S_i = +1$ and $\rho=0.25$, for example, rejecting $H_i^{D,\rho}$ amounts to an assertion that $\theta_{i,R} < 0.75 \,\theta_{i,O}$, i.e. the true effect declined by more than $25\%$, or is negative.

  In particular, if $\rho=0$ then $H_i^{D,0}$ is a one-sided version of $H_i^E$, and when $\rho=1$, $H_i^{D,1}$ is equivalent to $H_i^{S,R}$. We can subsequently ask how many of $H_i^{D,\rho}$ are false: another multiple testing problem. We provide two estimators (one overestimate and one underestimate) and confidence interval for the proportion of false $H_i^{D,\rho}$.

  To facilitate the rest of the paper, we recapitulate the notations above in \Cref{tbl:notation}.
  \begin{table}[htbp]
    \centering
    \begin{tabularx}{\textwidth}{lX}
      \toprule
      Symbol & Meaning \\
      \midrule
      $\theta_{i,O}$, $\theta_{i,R}$ & True effect sizes in the original and replication $i$-th studies \\
      $\htheta_{i,O}$, $\htheta_{i,R}$ & Effect size estimates in the original and replication $i$-th studies \\
      $\sigma_{i,O}$, $\sigma_{i,R}$ & Standard errors in the original and replication $i$-th studies \\
      $S_i$ & Sign of the original $i$-th study \\
      $H_i^{S,O}$, $H_i^{S,R}$ & Hypothesis that a type S error has occurred in the original and replication $i$-th studies  \\
      $H_i^E = H_i^{E,0}$, $H_i^{E,\delta}$ & Hypothesis that the true effect has shifted by $\delta$ from the original to the replication $i$-th study \\
      $H_i^{D,\rho}$ & Hypothesis that the true effect size has declined by no more than the fraction $\rho$ \\
      \bottomrule
    \end{tabularx}
    \caption{Summary of notations introduced.}
  \label{tbl:notation}
  \end{table}

\subsection{Data-dependent hypotheses and conditional inference}

  Our hypotheses above, $H_i^{S,O}$, $H_i^{S,R}$ and $H_i^{D,\rho}$, are all innately data-dependent. While data-dependent hypotheses may at first sound unusual, they are commonplace in practice, for example when pilot studies are performed to generate hypotheses that are tested later on with fresh data. There is no inherent conceptual problem with testing these data-dependent hypotheses: intuitively, we understand that the test remains valid because the type I error rate is controlled for whatever hypothesis is selected, conditional on that hypothesis having been selected.

  Conditional inference is well-established in the statistical literature as a means of constructing valid confidence intervals for parameters that were selected in a data-dependent way \citep[e.g.][]{Sampson:2005,Zollner:2007,Weinstein:2013,Yekutieli:2012}. \citet{Fithian:2014ws} generalized the intuition about pilot studies to argue that a test of a data-dependent hypothesis is valid, so long as the type I error rate is controlled conditioned on the portion of the data that generated the hypothesis. For our hypotheses here, $S_i$ is the part of the data that determines the hypothesis: in effect, we can imagine ourselves in the position of having observed the signs of all the original estimators, but knowing nothing else about the data. At that stage, it is valid to formulate a hypothesis that depends on $S_i$, and plan to test it using the still-unobserved data: namely, $|\htheta_{i,O}|$ and $\htheta_{i,R}$.

  After conditioning on $S_i$, each hypothesis discussed above amounts to testing a fixed linear hypotheses about $(\theta_{i,R}, \theta_{i,O})$, the natural parameter of the truncated bivariate normal model~\eqref{eq:trunc-model}; as a result, they are all amenable to post-selection inference using the selective $z$-test built on the work of \citet{Lee:2016fv}. \Cref{sec:methodology} discusses the methodology in detail.

\subsection{Related work}

  There has been much commentary on how to define replicability for scientific experiments. \citet{Valentine:2011ga} pointed out that the definition should depend on the scientific context. For example, sometimes one may wish to test the robustness of conclusions to subpopulation differences, but in other times, to changes in experimental conditions. \citet{Goodman:2016bo} expanded on this, and gave a few useful definitions for what replicability is, such as {\em methods reproducibility}, {\em results reproducibility}, {\em inferential reproducibility}, etc., but stopped short of an operational statistical criterion for replicability. False directional claims and effect shift can be loosely interpreted as inferential and results reproducibility, respectively.

  Operationally, \citet{Valentine:2011ga} and \citet{Nosek:2017ei} proposed the metrics used in RP:P and \citet{Camerer:2018de}, a similar replication effort in experimental economics. These metrics however suffer the shortcomings discussed earlier, in that they do not answer a concrete statistical question and cannot disentangle selection bias from other explanations.

  In this article, our definitions of replicability are inspired primarily by the statistical literature on multiple testing and meta-analysis, such as the estimator in \citet{Storey:2002vj}, the FDP and directional FDP from \citet{Benjamini:2000ka,Benjamini:2005bv}, and the partial conjunction testing framework of \citet{Heller:2007,Benjamini:2008kj}. Related error rates have also been estimated before: \citet{Jager:2013dq} have modeled the $p$-value distributions under alternatives and the selection for statistical significance to estimate the FDR in the medical literature, accompanied by useful discussions from \citet{Gelman:2013ep,Goodman:2013kj,Ioannidis:2013fz}; in addition, \citet{Camerer:2018de} used Bayesian methods to estimate the false positive rate, instead of the FDR, for published social science results in {\em Nature} and {\em Science}.

  Furthermore, there are many past efforts to model and quantify selection bias, specifically using the RP:P dataset. For instance, \citet{Johnson:2017gu} considered a publication bias model where the probability of publication is a step function of the $p$-value, which is generalized nonparametrically in \citet{Andrews:2018vh}. The two analyses estimated that a statistically significant result was $200$ \citep{Johnson:2017gu} or $30$ \citep{Andrews:2018vh} times as likely to be published as a statistically insignificant one.

  Adjusting for selection, \citet{vanAert:2017,vanAert:2018} have combined the evidences from both the original and replication experiments to provide estimates for the effect sizes. Specifically with a truncated Gaussian model, \citet{Etz:2016gx} have also analyzed the RP:P dataset from a Bayesian perspective, and investigated the discrepancies between the original and replication studies. Our analysis provides a complementary point of view with frequentist hypothesis testing without any prior on the effect sizes, with the help of recent advances in post-selection inference, including primarily the selective $z$-test framework of \citet{Lee:2016fv}.

\subsection{Outline}

  \Cref{sec:methodology} details the methodology and assumptions used in this analysis, and is somewhat technical. \Cref{sec:analysis} applies the developed methodology to the RP:P dataset, summarizes and interprets the results. \Cref{sec:discussion} concludes.

\section{Methodology}
\label{sec:methodology}

  In this section we will construct an estimator for directional FDP, a test for the effect shift in replication $i$ and an estimator for the proportion of effect sizes that declined. We also use $X \gest Y$ to denote that $X$ is stochastically larger than $Y$. The index $i$ is suppressed when there is no risk of ambiguity. 

  Since we need a well-defined notion of direction to consider the proportion of false directional claims, we restrict our attention to univariate tests, namely $z$-, $t$-, $F(1, \cdot)$-tests or correlations. Thus, studies that are not univariate or have $p$-values greater than $\ao = 0.05$ are discarded: our estimates and analyses below consider only the $m = 68$ remaining studies with univariate structure and conventionally significant original $p$-values.

\subsection{Selection bias model}
  Model~\eqref{eq:trunc-model} assumes that results are only published if they achieved statistical significance at some conventional threshold level $\ao$, which is $0.05$ in our data. While this assumption is not true in the case of RP:P since some original $p$-values are above $0.05$, we note that the model can be relaxed to the following milder assumption:
  \begin{assumption}
    \label{ass:sig-enough}
    $p_O < \ao$ is ``significant enough'': a result with $p_O < \ao$ would be equally likely to be published (or selected for replication), if the $p$-value were some other statistically significant value.
  \end{assumption}

  \Cref{ass:sig-enough} allows some significant $p$-values to go unpublished. If it holds, then we can model the original test statistics as following their theoretical distribution, truncated to the event where the corresponding $p$-values are below $\ao$, as in Model~\ref{eq:trunc-model}.

  Note that \Cref{ass:sig-enough} contemplates a fairly straightforward mechanism for selection on statistical significance, which may not be adequate to describe the effects of more complex and difficult-to-model QRPs. In particular, $p$-hacking --- the iterative tweaking of an analysis until the $p$-value drops below the researcher's desired significance level $\ao$ --- is commonly suspected to produce a pileup of $p$-values just below the significance threshold \citep[see e.g.][]{Simonsohn:2014fa}. Because $p$-hacking is such a vaguely defined practice, it is unclear how we might incorporate it into our model, but in any case there is no evidence of a pileup just below $0.05$ in the original RP:P studies (see \Cref{fig:fdp-original}). We will reconsider the validity of \Cref{ass:sig-enough} in \Cref{sec:discussion}.

\subsection{False directional claims}

  We will adapt the method in \citet{Storey:2002vj} to estimate the directional FDP while accounting for selection bias. Furthermore, if we believe the chosen studies are representative of the publications in the journal or discipline \citep[e.g.][]{Stroebe:2016dj}, then this estimator can also be regarded as an estimator for the journal-wide or discipline-wide directional false discovery rate (FDR\textsubscript{dir}), the expectation of the directional FDP \citep{Benjamini:2005bv}.

  \paragraph{Adjusting for selection bias} While dividing a post-selection $p$-value by $\ao$ intuitively adjusts for selection, it is not immediately valid when the null is one-sided with a true effect not on the boundary. We demonstrate below that this adjustment typically remains valid even in this case.

  Recall that a valid $p$-value is a random variable that is stochastically larger than $\Uniform[0, 1]$ (i.e.\ superuniform) under the null hypothesis. If we only observe the original $p$-value when it is significant, it is not superuniform after selection under $H^{S,O}$, and it is therefore not valid for testing the hypothesis of a false directional claim. To adjust these $p$-values for selection, we follow the principle in \citet{Fithian:2014ws} by conditioning on the event that the $p$-values are selected, and also on the variable $S = \sgn(\htheta_O)$ which determines the hypothesis $H^{S,O}$ that we test. We consider two cases: when the original study is a one-sided test and when it is a two-sided test. As we will see, the adjustment in either case is to divide by $\ao$.

  First we consider the case where the original study was a one-sided test. Assume $p_O$ is a $p$-value for a test of the hypothesis $H_0:\; \theta_O \le 0$, in which case $S = +1$ deterministically (the opposite case with $H_0:\; \theta_O \ge 0$, and $S = -1$ deterministically, is directly analogous). Suppose $p_O$ is the original $p$-value, which we observe only when it is significant at the conventional threshold, i.e. when $p_O < \ao$. Under mild assumptions satisfied by both $z$-tests and $t$-tests,\footnote{namely, that the test statistic has monotone likelihood ratio in the parameter} $p_O \gest \Uniform[0, \ao]$ under $H^{S,O}$, in which case $p_O/\ao \gest \Uniform[0,1]$.

  Next we consider the case where $p_O$ is a $p$-value for a two-sided test of $H_0:\; \theta_O = 0$, and where $S = +1$ (the case with $S = -1$ is analogous). If $p_O^+$ was the original one-sided $p$-value for $H_0:\; \theta_O\le 0$, then $p_O = 2p_O^+$ when $S=+1$ ($p_O = 2 - 2p_O^+$ if $S = -1$). In our truncated model, under the same assumptions as above and conditional on $S = +1$, $p_O^+ \gest \Uniform[0, \ao/2]$ and therefore $p_O/\ao = 2p_O^+/\ao \gest \Uniform[0, 1]$ under $H^{S,O}$. We write $p_O'=p_O/\ao$ for the adjusted $p$-value.

  \paragraph{Inference on FDP: estimate and upper confidence bound} Using the adjusted original $p$-values, we can estimate the directional FDP in the original studies. Recall from \Cref{tbl:err} that
  \begin{align*}
    R &= \#\{p_{i,O} \le \ao\} = m,\\
    V &= \#\{p_{i,O} \le \ao \text{ and } H_i^{S,O} \text{ is true}\}.
  \end{align*}
  Since all of the studies were deemed discoveries, $R=m$ is the total number of studies here. \Cref{tbl:err-fdp} classifies the $m$ conventionally significant studies according to whether $H_i^{S,O}$ is true and whether the adjusted $p$-value is larger than some fixed value $\lambda$ in $(0,1)$, e.g.\ $\lambda = 0.5$.
  \begin{table}[htbp]
    \centering
    \begin{tabular}{lccc}
      \toprule
      Adjusted $p$-value & $H_i^{S,O}$ is true & $H_i^{S,O}$ is false & Total \\
      \midrule
      $p_{i,O}' < \lambda$ & $*$ & $*$ & $*$ \\
      $p_{i,O}' \ge \lambda$ & $U$ & $*$ & $B$ \\
      Total & $V$ & $*$ & $R = m$ \\
      \bottomrule
    \end{tabular}
    \caption{Classification of the $R=m$ significant original studies. Here only $R$ and $B$ are observed, and we wish to infer on $V$.}
  \label{tbl:err-fdp}
  \end{table}

  Note that $B = \#\{\lambda\ao \le p_{i,O} < \ao\}$ from \Cref{tbl:err-fdp} is observable, while $V$ and $U$ are not. Under the one-sided null, the $p$-value is superuniform, and so 
  \begin{equation}
    B \gest U \gest \Binomial(V, 1-\lambda).
  \label{eq:Bgest}
  \end{equation}
  As a result, $\EE[B] \ge (1-\lambda)V$ and a conservative (upwardly biased) estimator of the directional FDP is
  \[
    \hFDPd = \frac{B}{(1 - \lambda)R}.
  \]
  This estimate is conservative in the sense that it overestimates the type I error, and is equivalent to the estimator $\hat{\pi}_0$ of the true null proportion in \citet{Storey:2002vj}. Using $\lambda = 0.5$ and $\ao = 0.05$, the estimate boils down to
  \[
    \hFDPd = \frac{2}{m} \cdot \#\{0.025 \le p_{i,O} < 0.05\}.
  \]

  While the above is formally an estimator for the number of directional errors, it can be interpreted practically as an estimate of the fraction of directional claims where {\em either} the direction is wrong {\em or} the effect has a negligible magnitude, cf.\ type M error from \citet{Gelman:2014cd}. This is because $p$-values whose effect sizes are very close to zero are nearly uniform and contribute to our estimator similarly as if the true effect were exactly zero.

  Additionally, we can exploit \eqref{eq:Bgest} to obtain an upper confidence bound for the directional FDP, by testing the hypothesis $H_0:\; V \ge v_0$, a partial conjunction hypothesis investigated in \citet{Heller:2007}. Here we combine only the coarse information of whether each $p$-value is greater than $\lambda$,\footnote{More precisely, we count number of $p$-values that are greater than $\lambda$ and consider its distribution under the partial conjunction null hypothesis} and reject for small values of $B$. We can compute the largest $v_0$ such that the test still accepts, which gives an upper confidence bound of $V$. Dividing this bound by $R$ gives an upper confidence bound for the directional FDP.

  \paragraph{Directional FDP at smaller thresholds} One proposal to address the replicability crisis is to lower the conventional significance threshold from $\ao=0.05$ to some smaller value $\alpha$, such as $0.005$ \citep{Benjamin:2018gh}. As suggested by \citet{Goodman:2013kj}, an empirical method to evaluate the hypothetical scenario with a smaller threshold can be helpful. We now discuss methods for inference on the directional FDP for those studies with $p_O < \alpha < \ao$, based on comparing the number of adjusted $p$-values below $\alpha$ with the number above $\lambda\ao$, for some $\lambda > \alpha/\ao$. We call this method the {\em external comparison method} in contrast to the earlier {\em internal comparison method} that bases on \eqref{eq:Bgest}. This method will be less conservative as we are not constrained to only using the $p$-values in $[0, \alpha)$.
  
  Let $N \le m$ denote the total number of original $p$-values in $[0,\alpha) \cup [\lambda\ao, \ao)$ (or equivalently, the number of adjusted $p$-values in $[0,\alpha') \cup [\lambda, 1)$ for $\alpha' = \alpha/\ao$). \Cref{tbl:big-small} classifies these $N$ studies according to whether $H_i^{S,O}$ is true and whether the adjusted $p$-value is larger than $\lambda$ or smaller than $\alpha'$. The numbers of false directional claims and all directional claims under the hypothetical threshold are $V_\alpha$ and $R_\alpha$, respectively. Auxiliary counts, $T_\alpha$ and $W$, are defined according to \Cref{tbl:big-small} as well. The directional FDP, $V_\alpha / R_\alpha$, remains as our quantity of interest.
  
  \begin{table}[htbp]
    \centering
    \begin{tabular}{lccc}
      \toprule
      Adjusted $p$-value & $H_i^{S,O}$ is true & $H_i^{S,O}$ is false & Total \\
      \midrule
      Small ($p_{i,O}' < \alpha'$) & $V_\alpha$ & $T_\alpha$ & $R_\alpha$ \\
      Big ($p_{i,O}' \ge \lambda$) & $U$ & $W$ & $B$ \\
      Total & $N_0$ & $*$ & $N$ \\
      \bottomrule
    \end{tabular}
    \caption{Classification of the $N \le m$ original studies with adjusted $p$-values in $[0, \alpha'] \cup [\lambda, 1]$. Only $R_\alpha$, $B$ and $N$ are observed. Auxiliary unobserved quantities, $N_0$, $T_\alpha$ and $R_\alpha$, are defined accordingly. Our goal is to infer on $V_\alpha$.}
    \label{tbl:big-small}
  \end{table}
  
  Our method is inspired by the following stochastic inequality.
  \begin{lemma}
    Conditional on $N$, $T_\alpha$ and $W$, we have
    \begin{equation}
      B \mid N, T_\alpha, W \gest \Binomial(N - T_\alpha, \beta).
    \label{eq:B-NTW}
    \end{equation}
  \end{lemma}
  \begin{proof}
    All adjusted $p$-values are independent, and are either small ($p \le \alpha'$) or big ($p \ge \lambda$). The adjusted $p$-values corresponding to a true null are big with probability at least $\beta = \frac{1 - \lambda}{1 - \lambda + \alpha'}$. We proceed to condition on $T_\alpha$ and $W$, so they are now considered deterministic. So the total number of big adjusted $p$-values, $B$, satisfies
    \[
      B = U + W \gest \Binomial(N-N_0, \beta) + W \gest \Binomial(N-T_\alpha, \beta).
    \]
  \end{proof}

  With \eqref{eq:B-NTW}, we can estimate $N - T_\alpha$ conservatively with $B / \beta$. Since $V_\alpha  = N - T_\alpha - B$, a reasonable estimator for the directional FDP is
  \[
    \hFDPd = \frac{1 - \beta}{\beta} \cdot \frac{B}{R_\alpha}.
  \]
  Furthermore \eqref{eq:B-NTW} gives us a $95\%$ upper confidence bound for the directional FDP:
  \[
    \FDPd^* = \frac{Q - B}{R_\alpha}, \qquad\text{where } Q = \max\{q: \PP[\Binomial(q, \beta) \ge B] \ge 0.95\}.
  \]

  \begin{proposition}
    The expectation of $\hFDPd$ is at least the expectation of the true directional FDP, and $\FDPd^*$ is greater than the true directional FDP, with probability at least $95\%$.
    \label{prop:prereg-test}
  \end{proposition}
  \begin{proof}
    For the estimator, we start by taking the expectation of $\hFDPd - \FDPd$, conditional on $N$, $T_\alpha$ and $W$:
    \begin{align}
      \EE[\hFDPd - \FDPd \mid N, T_\alpha, W] &= \EE\left[\frac{\frac{1-\beta}{\beta} B - V_\alpha}{R_\alpha} \,\middle|\, N, T_\alpha, W\right] \nonumber \\
      &\ge \EE\left[\frac{\frac{1-\beta}{\beta} (N_0 - V_\alpha) - V_\alpha}{V_\alpha + T_\alpha} \,\middle|\, N, T_\alpha, W\right] \nonumber \\
      &= \EE\left[\frac{(1-\beta) N_0 - V_\alpha}{\beta (V_\alpha + T_\alpha)} \,\middle|\, N, T_\alpha, W\right] \nonumber \\
      &\ge \frac{(1-\beta) N_0 - \EE[V_\alpha \mid N, T_\alpha, W]}{\beta (\EE[V_\alpha \mid N, T_\alpha, W] + T_\alpha)} \label{eq:jensen} \\
      &\ge 0 \label{eq:binom-ineq},
    \end{align}
    where \eqref{eq:jensen} follows from applying Jensen's inequality to the convex function $f(x) = \frac{(1-\beta) N_0 - x}{\beta(x + T_\alpha)}$, and \eqref{eq:binom-ineq} follows from $V_\alpha \mid N, T_\alpha, W \lest \Binomial(N_0, 1-\beta)$. Taking expectation on both sides completes the proof.

    For $\FDPd^*$, we can directly compute the probability that it is greater than $\FDPd$, conditional on $N$, $T_\alpha$ and $W$:
    \begin{align*}
      \PP[\FDPd^* \ge \FDPd \mid N, T_\alpha, W] &= \PP\left[\frac{Q - B}{R_\alpha} \ge \frac{V_\alpha}{R_\alpha} \,\middle|\, N, T_\alpha, W\right] \\
      &= \PP[Q \ge B + V_\alpha \mid N, T_\alpha, W] \\
      &= \PP[Q \ge N - T_\alpha \mid N, T_\alpha, W] \\
      &\ge 0.95,
    \end{align*}
    from the construction of $Q$. Taking expectation on both sides hence yields the desired marginal coverage.
  \end{proof}

  \begin{remark}
    This proof of conservativeness actually shows something stronger than marginal guarantees: the estimator and confidence upper bound are both conservative conditionally, even when we condition on the signs $S_i$.
  \end{remark}

  \paragraph{Methods using replication $p$-values} As mentioned in \Cref{sec:intro}, we can use the replication $p$-values in lieu of the adjusted original $p$-values above, providing an estimate and confidence bound for the frequency of when the $\htheta_O$ incorrectly predicts the replication effect direction. While this approach requires potentially costly replications in future applications, it provides valuable additional information. In particular, the replication $p$-values are more likely to be free of QRPs or $p$-hacking that may violate our assumption that adjusted $p$-values are superuniform under the null, providing more robust evidence regarding replicability. The corresponding estimator for unadjusted replication $p$-values with $\lambda= 0.5$ is 
  \[
    \hFDPd = \frac{2}{m} \cdot \#\{p_{i,R} \ge 0.5\}.
  \]

\subsection{Effect shift}

  We will derive a test for the hypothesis $H^E: \theta_O = \theta_R$ at level $0.05$. Our test is based on a normal distribution, so we start by demonstrating that the effect size estimates of the univariate studies can be reasonably modeled by our truncated bivariate normal distribution in model~\eqref{eq:trunc-model}. We classify these studies into two categories and provide a rough rationale in our definition of effect size in each category:
  \begin{inlinelist}
    \item $t$-tests and $F(1, \cdot)$ ANOVAs, where all independent variables are categorical; and,
    \item correlations and regressions, where one or more independent variables are continuous.
  \end{inlinelist}

  For a $t$-test or $F(1, \cdot)$ ANOVA, we can define the effect size as the noncentrality parameter, scaled for cell sizes. In other words, the $t$-statistic is distributed as $T \sim t_{df}(k\theta)$, for some real constant $k$ chosen based on the study design. For example, $k = \sqrt{n}$ for a one-sample $t$-test. When $df$ is sufficiently large, the $t$-statistic is approximated well by a $z$-statistic, and distributed approximately as
  \[
    T \sim N(k\theta, 1).
  \]
  For our analysis, we consider studies where the original and replication degrees of freedom are at least $30$.\footnote{The choice of $30$ complies with the analysis in \citet{Andrews:2018vh}. Further discussion on the approximation in available in the supplement.}

  For a (partial) correlation coefficient estimate, $R$, we can apply Fisher transformation \citeyearpar{Fisher:1921vq,Fisher:1924ve} to convert it into a $z$-statistic, which approximately follows
  \[
    \sqrt{n-3-p} \tanh^{-1}(R) \sim N(\sqrt{n-3-p}\,\theta, 1),
  \]
  where $p$ is the number of controlled covariates and $\theta$ is a quantity that can be taken as the effect size.

  In either case, the test statistic in $46$ studies can be transformed to an approximate $z$-score $Z \sim N(k \theta, 1)$ for some real constant $k$. Additional considerations in certain studies are detailed in the supplement.

  \paragraph{Adjusting for selection bias} We turn next to address the issue of post-selection inference. Again, we condition on the event where the $z$-scores are observed, but we do not need to condition on $S$ as the hypothesis $H^E$ is no longer random. Since the statistic is only observed if it is statistically significant, the original and replication $z$-statistics follow a truncated bivariate normal joint distribution:
  \[
    \begin{bmatrix}
      Z_O \\ Z_R
    \end{bmatrix} \sim N\left(\begin{bmatrix}
      k_O \theta_O \\ k_R \theta_R
    \end{bmatrix}, \begin{bmatrix}
      1 & 0 \\ 0 & 1
    \end{bmatrix}\right) 1_{\{Z_O \in A\}}.
  \]
  Here $A$ is the selection event, which contains the statistically significant values of $Z_O$. We are interested in testing $H^E: \theta_O = \theta_R$ and more generally the null hypothesis $H^{E,\delta}: \theta_O - \theta_R = \delta$, which can be inverted to yield a confidence interval.

  We cast this as a more general testing problem here to benefit later derivations on effect decline. Suppose we have a truncated bivariate distribution
  \[
    Z = \begin{bmatrix}
      Z_1 \\ Z_2
    \end{bmatrix} \sim N\left(\mu, \begin{bmatrix}
      1 & 0 \\ 0 & 1
    \end{bmatrix}\right) 1_{\{Z_1 \in A\}}, \quad \text{where } \mu = \begin{bmatrix}
      \mu_1 \\ \mu_2
    \end{bmatrix},
  \]
  and we want to test $\eta' \mu = \delta$ for some constant vector $\eta = (\eta_1, \eta_2)$ with $\eta_1 > 0$. Test for $H^E$ and $H^{E,\delta}$ are special cases where $\eta = (1/k_O, -1/k_R)$.

  We can perform this general testing problem with the {\em selective $z$-test}, based on the framework in \citet{Lee:2016fv}.

  \begin{definition}[Selective $z$-test]
  \label{def:sel-gauss}
    Let $\eta_\perp = (\eta_2, -\eta_1)$, $D = \eta' Z$ and $M = \eta_\perp' Z$. We now consider $M$ as a constant and test $\eta' \mu = \delta$ using the test statistic $D$ against the null distribution
    \[
      N(\delta, \|\eta\|^2) 1_{\left\{D \in \frac{\|\eta\|^2 A - \eta_2 M}{\eta_1}\right\}}.
    \]
    Specifically, we reject $\eta' \mu = \delta$ when $D$ is below the $\frac{0.05}{2}$-quantile or over the $(1 - \frac{0.05}{2})$-quantile of this null distribution.
  \end{definition}

  We proceed to show that this is a valid test by construction.

  \begin{proposition}
  \label{prop:sel-gauss}
    The selective $z$-test defined in \Cref{def:sel-gauss} has level $0.05$.
  \end{proposition}

  \begin{proof}
    Leveraging the fact that $\eta' \eta_\perp = 0$, we reparametrize the joint distribution of $(Z_1, Z_2)$ under the null such that $\delta$ is a parameter, i.e.
    \[
      \begin{bmatrix} D \\ M \end{bmatrix} =
      \begin{bmatrix} \eta' Z \\ \eta_\perp' Z \end{bmatrix}
      \sim N\left(
      \begin{bmatrix} \delta \\ \eta_\perp' \mu \end{bmatrix},
      \begin{bmatrix}
        \|\eta\|^2 & 0 \\
        0 & \|\eta\|^2
      \end{bmatrix}\right) 1_{\{Z_1 \in A\}}.
    \]
    In particular, the event $Z_1 \in A$ can be rewritten as
    \[
      D \in \frac{\|\eta\|^2 A - \eta_2 M}{\eta_1}.
    \]
    And so the distribution of $D$ conditional on $M$ under $H_0^\delta$ is a truncated Gaussian distribution,
    \[
      [D \mid M] \sim N\left(\delta, \|\eta\|^2\right) 1_{\left\{D \in \frac{\|\eta\|^2 A - \eta_2 M}{\eta_1}\right\}}
    \]
    and we obtain a valid test by rejecting when $D$ is smaller than the $\frac{0.05}{2}$-quantile or larger than the $\left(1 - \frac{0.05}{2}\right)$-quantile.
  \end{proof}

  The construction above is represented graphically in \Cref{fig:sel-gauss}, in the style of \citet{Lee:2016fv}. We can represent the observation $(Z_1, Z_2)$ as a point in $\mathbb{R}^2$. Conditioning on $M$ is equivalent to conditioning on $M / \|\eta_\perp\|$, which means we are now considering the conditional distribution on the truncated line $\ell$. The test statistic $D$, or equivalently $D / \|\eta\|$, indicates the position on $\ell$. Under the null that $\eta' \mu  = \delta$, the conditional distribution on $\ell$ is known and a valid $p$-value can be obtained, yielding the selective $z$-test.
  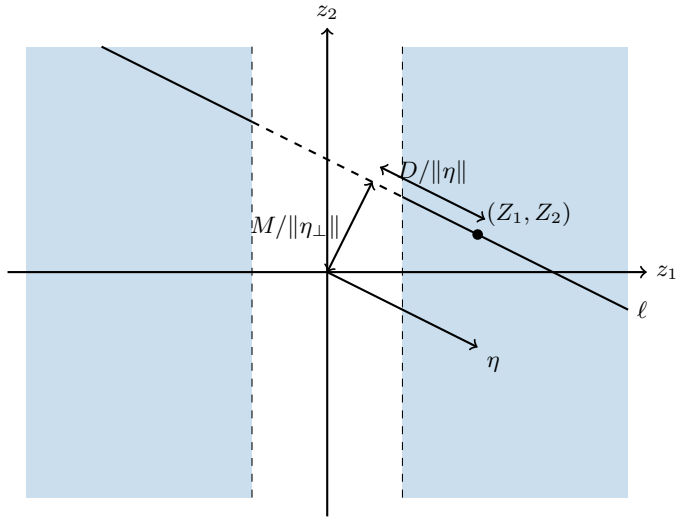
\begin{figure}[htbp]
    \centering
    \begin{tikzpicture}
    \fill[mid-blue] (-4, -3) -- (-4, 3) -- (-1, 3) -- (-1, -3) -- cycle;
    \fill[mid-blue] (4, -3) -- (4, 3) -- (1, 3) -- (1, -3) -- cycle;
    \draw[thick, ->] (0, -3.25) -- (0, 3.25) node[above] {$z_2$};
    \draw[thick, ->] (-4.25, 0) -- (4.25, 0) node[right] {$z_1$};
    \draw[dashed] (-1, -3) -- (-1, 3);
    \draw[dashed] (1, -3) -- (1, 3);
    \fill (2, 0.5) circle(2pt) node[above right] {$(Z_1, Z_2)$};
    \draw[thick, ->] (0, 0) -- (2, -1) node[below right] {$\eta$};

    \draw[thick, <->] (0, 0) -- (0.6, 1.2) node[midway, left] {$M / \|\eta_\perp\|$};
    \draw[thick] (-3, 3) -- (-1, 2);
    \draw[thick, dashed] (-1, 2) -- (1, 1);
    \draw[thick] (1, 1) -- (4, -0.5) node[right] {$\ell$};

    \draw[thick, <->] (0.7, 1.4) -- (2.1, 0.7) node[midway, above] {$D / \|\eta\|$};
    \end{tikzpicture}
    \caption{Graphical representation of the selective $z$-test. The observation $(Z_1, Z_2)$ is a point and the truncation on $Z_1$ means that the shaded area is the support of the joint distribution $(Z_1, Z_2)$. Conditioning on $M$ is the same as conditioning on $M / \|\eta_\perp\|$, so we now consider the conditional distribution on the truncated line $\ell$. The test statistic $D$ indicates the position on $\ell$. Under the null $H^{E,\delta}: \theta_1 - \theta_2 = \delta$, the conditional distribution on $\ell$ is known and a valid $p$-value can be obtained, yielding the selective $z$-test.}
  \label{fig:sel-gauss}
  \end{figure}

  \begin{remark}
    It is not necessary to use $\frac{0.05}{2}$- and $(1 - \frac{0.05}{2})$-quantiles of the null distribution, as long as the desired significance level is achieved under the null distribution. For example, a uniformly most powerful unbiased test can be used in lieu of a test with equal tail cutoffs. Furthermore, if we are interested in a one-sided hypothesis, e.g.\ $\eta' \mu \le 0$, we can reject on one tail only. This will be particularly useful for derivations about effect decline later.
  \end{remark}

  \paragraph{Interval estimation} Given a valid test $\phi(Z_O, Z_R)$ for testing $H^{E,\delta}: \theta_O - \theta_R = \delta$, we can obtain two intervals: a predictive interval for the replication effect size estimate, and a confidence interval for effect shifts.

  Under the null hypothesis $H^E: \theta_O = \theta_R$, $\PP[\phi(Z_O, Z_R) \text{ rejects}] = 0.05$, or equivalently,
  \[
    \PP[\{z_R: \phi(Z_O, z_R) \text{ accepts}\} \ni Z_R] = 0.95.
  \]
  Hence $\{z_R: \phi(Z_O, z_R) \text{ accepts}\}$ is a predictive interval for $Z_R$, which translates to a predictive interval for the point estimate $\htheta_R$ of the replication effect size.

  By the duality of hypothesis testing and confidence set, the set
  \[
    \{\delta: H^{E,\delta} \text{ is rejected}\}
  \]
  covers the difference of the original and replication effect sizes with probability $95\%$.

\subsection{Effect decline}

  We will estimate the proportion of effect sizes that declined by at least a fraction of $\rho$. Our procedure consists of two parts:
  \begin{inlinelist}
    \item for each study $i$, test and produce a $p$-value for the hypothesis $H_i^{D,\rho}$, and
    \item adapt the method for the directional FDP to estimate the proportion of $H_i^{D,\rho}$ that are false.
  \end{inlinelist}

  \paragraph{Adjusting for selection bias} As with the exactness test, we condition not only on the event where the $z$-scores are observed, but also on $S = \sgn(\htheta_O)$ as our hypothesis $H^{D,\rho}$ is determined by this random variable. In other words, we consider the $z$-statistic $Z_O$ to be drawn from the set $A_+$, where $A$ is the selection event from our test for effect shift and
  \[
    A_+ = A \cap \RR_+ = \{z_O: z_O \text{ is statistically significant}\} \cap \RR_+.
  \]
  Putting $Z_O$ and $Z_R$ together, they follow a truncated bivariate normal joint distribution:
  \[
    \begin{bmatrix}
      Z_O \\ Z_R
    \end{bmatrix} \sim N\left(\begin{bmatrix}
      k_O \theta_O \\ k_R \theta_R
    \end{bmatrix}, \begin{bmatrix}
      1 & 0 \\ 0 & 1
    \end{bmatrix}\right) 1_{\{Z_O \in A_+\}}.
  \]
  By convention RP:P chose $\htheta_O > 0$ so the hypothesis $H^{D,\rho}$ reduces to $\theta_{i,R} \ge (1-\rho) \theta_{i,O}$, or equivalently $\theta_{i,R} - (1-\rho) \theta_{i,O} \ge 0$. This can be tested using the selective $z$-test with $\eta = (1/k_O, -1/(1-\rho)k_R)$ and rejecting on one tail only.

  \paragraph{Inference on effect decline: estimates and confidence bounds} With the resulting $p$-values, our earlier methods on directional FDP can provide an overestimate and a upper confidence bound for the proportion of true $H^{D,\rho}$. Subtracting these from $1$ yields an underestimate and a lower confidence bound for the proportion of false $H^{D,\rho}$. On the other hand, by considering the complement of the hypothesis $H^{D,\rho}$, we can also provide an overestimate and an upper confidence bound for the proportion of false $H^{D,\rho}$. These estimators and bounds together provide an overestimate, an underestimate and a $90\%$ confidence interval for the proportion of effect sizes that at least declined by a fraction of $\rho$.

\section{Re-analysis of RP:P}
\label{sec:analysis}

\subsection{False directional claims}

  We implemented our method with $\lambda = 0.5$ to estimate the number of one-sided nulls and the directional FDP.\footnote{Choosing $\lambda = 0.5$ follows the convention in the multiple testing literature for a bias-variance trade off: if $\lambda$ is too small, many true discoveries are counted as false; if $\lambda$ is too big, the estimator can have large variance.} The adjusted original $p$-values and replication $p$-values are given in \Cref{fig:fdp-original,fig:fdp-replication} respectively. Using the original $p$-values, we estimate that $22$ of the $68$ ($32\%$) original directional claims are false, with a $95\%$ upper confidence bound of $47\%$. Using the replication $p$-values, we estimate that $32$ of the $68$ ($47\%$) original directional claims incorrectly predict the direction of the replication effect, with a $95\%$ upper confidence bound of $63\%$. In particular both of our FDP estimates are much lower than the $64\%$ which could be suggested by a naive reading of RP:P \citep[e.g.][]{Baker:2015kd}. These numbers are summarized again in \Cref{tbl:fdp-sim} later. Furthermore, while we can compute a lower confidence bound, it will always be $0\%$ as the data is obviously consistent with many null hypotheses being slightly false.
  \begin{figure}[htbp]
    \centering
    \begin{subfigure}[t]{0.49\hsize}
      \centering
      \includegraphics[width=\hsize]{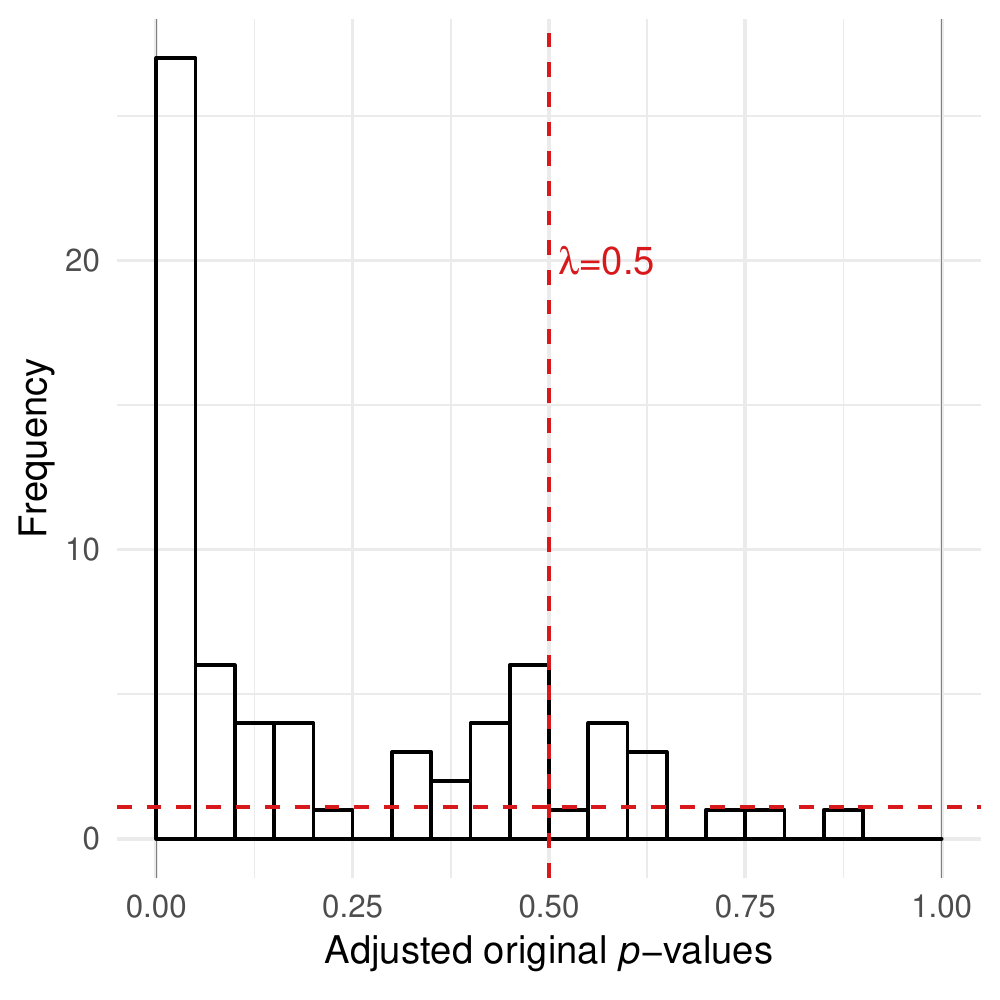}
      \caption{Histogram of the adjusted original $p$-values.}
    \label{fig:fdp-original}
    \end{subfigure}
    \hfill
    \begin{subfigure}[t]{0.49\hsize}
      \centering
      \includegraphics[width=\hsize]{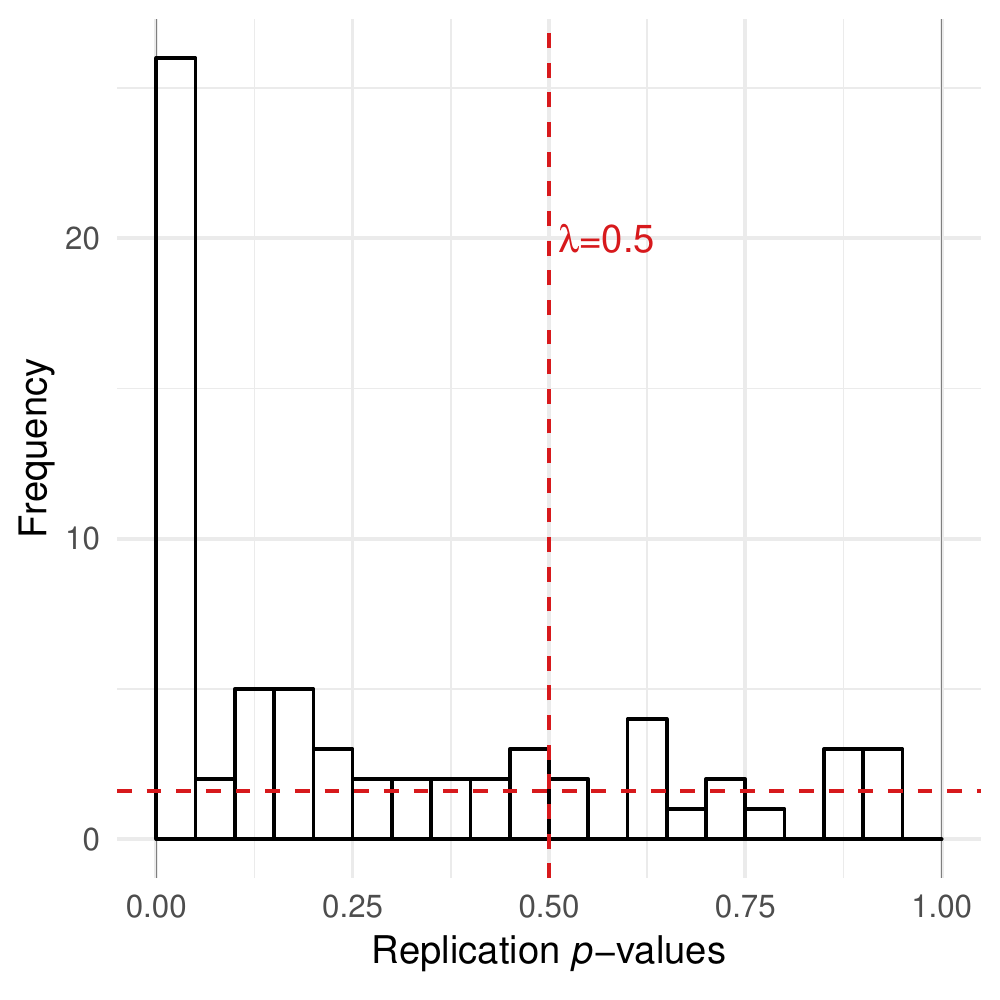}
      \caption{Histogram of the replication $p$-values.}
    \label{fig:fdp-replication}
    \end{subfigure}
    \caption{Method from \citet{Storey:2002vj} as demonstrated using histograms of $p$-values. We estimate the number of true nulls by conservatively assuming that every hypothesis right of the vertical red line to be true. Since the $p$-value under the null is superuniform, on average there are fewer null hypotheses left of the line than right of the line. Our overestimate of the number of true nulls in each bin is shown by the horizontal red line. A net excess of $p$-values above this line means false directional claims.}
  \end{figure}

  We proceeded to evaluate the proposal to reduce the statistical significance threshold \citep{Benjamin:2018gh}. We considered three candidates for the new threshold, $0.001$, $0.005$ and $0.01$, using the external comparison method. The directional FDP estimates and upper confidence bounds are given in \Cref{tbl:fdp-sim}.
  \begin{table}[htbp]
    \centering
    \begin{tabular}{rr@{\hspace*{0.25em}}r@{}lr@{\hspace*{0.25em}}r@{}lr@{\hspace*{0.25em}}rr@{\hspace*{0.25em}}r}
      \toprule
      \multirow{2}{*}{\raisebox{-\heavyrulewidth}{$\alpha$}} & \multicolumn{6}{c}{Adjusted original} & \multicolumn{4}{c}{Replication} \\
      \cmidrule(lr){2-7} \cmidrule(lr){8-11}
      & \multicolumn{3}{c}{Est.} & \multicolumn{3}{c}{U.C.B.} & \multicolumn{2}{c}{Est.} & \multicolumn{2}{c}{U.C.B.} \\
      \midrule
      $0.001$ & $0.4 / 22 =$&$2\%$&\textdagger & $2 / 22 =$&$9\%$&\textdagger & $6 / 22 =$&$27\%$ & $12 / 22 =$&$55\%$ \\
      $0.005$ & $2.2 / 33 =$&$7\%$&\textdagger & $6 / 33 =$&$18\%$&\textdagger & $12 / 33 =$&$36\%$ & $20 / 33 =$&$61\%$ \\
      $0.01$ & $4.4 / 41 =$&$11\%$&\textdagger & $9 / 41 =$&$22\%$&\textdagger & $16 / 41 =$&$39\%$ & $25 / 41 =$&$61\%$ \\
      $0.05$ & $22 / 68 =$&$32\%$ && $32 / 68 =$&$47\%$ && $32 / 68 =$&$47\%$ & $43 / 68 =$&$63\%$ \\
      \bottomrule
    \end{tabular}
    \caption{The directional FDP estimates and $95\%$ upper confidence bounds, using the adjusted original and replication $p$-values. The statistical significance level is $\alpha$. The external comparison method was used for computing the directional FDP estimates and the upper confidence bounds marked with daggers(\textdagger) above, as information of $p$-values between $\alpha$ and $0.05$ can improve the precision. The estimates and upper confidence bounds in the ``Replication'' column are relatively noisy due to the small number of $p$-values below the stricter rejection thresholds, and give little basis for any conclusions.}
  \label{tbl:fdp-sim}
  \end{table}

  These estimates corroborate \citet{Benjamin:2018gh}'s suggestion that reducing the statistical significance threshold may improve replicability, at least regarding the directional FDP of the original statistical hypotheses (of course, there is no way to account for potential change in researcher's behavior in response to the lowered threshold). Shall this be of interest, this method provides an empirical way to determine a better significance threshold, as no replications are needed. Nonetheless, potential effect heterogeneity is often a bigger concern. In this case, we are more concerned about the directional FDP for replications, which remains unacceptably high and requires replication experiments. Note, however, that a replication with low power could contribute to our estimates, even if there were no type S error.

\subsection{Effect shift}

  We performed the selective $z$-test for the hypothesis $H^E: \theta_O = \theta_R$ while adjusting for selection, where seven ($15\%$) studies are rejected. In contrast, without adjusting for selection, $18$ ($39\%$) studies are rejected at $0.05$ significance. If we wish to correct for multiplicity, we can apply Benjamini--Hochberg procedure \citeyearpar{Benjamini:1995cd}, which rules five ($11\%$) replication studies as inconsistent with the original studies at false discovery rate $0.10$.\footnote{The five rejected studies are \citet{Dodson:2008ks,vanDijk:2008br,PurdieVaughns:2008en,Farris:2008ev,Larsen:2008tu}.} Applying the more stringent Holm's method \citeyearpar{Holm:1979hl} to control the familywise error rate rules only the replication of \citet{Farris:2008ev} as inconsistent at familywise error rate $0.05$.

  We inverted the test for the hypothesis $H^E$, to yield a predictive interval for $Z_R$ and hence a predictive interval for the replication effect size estimate $\htheta_R$, shown in \Cref{fig:pi}. By definition $H^E$ is rejected when $\htheta_R$ is not included in the predictive interval. Adjusting for selection generally stretches the predictive intervals, resulting in fewer rejections.

  We also inverted the test for $H^{E,\delta}$ and obtained a confidence interval for the effect shifts, $\theta_O - \theta_R$, given in \Cref{fig:ci}. By construction the null hypothesis $H^E: \theta_O = \theta_R$ is rejected when the confidence interval does not include $0$. Adjusting for selection also generally lengthens the confidence intervals, resulting in fewer rejections.

  \begin{figure}[htbp]
    \centering
    \includegraphics[width=\hsize]{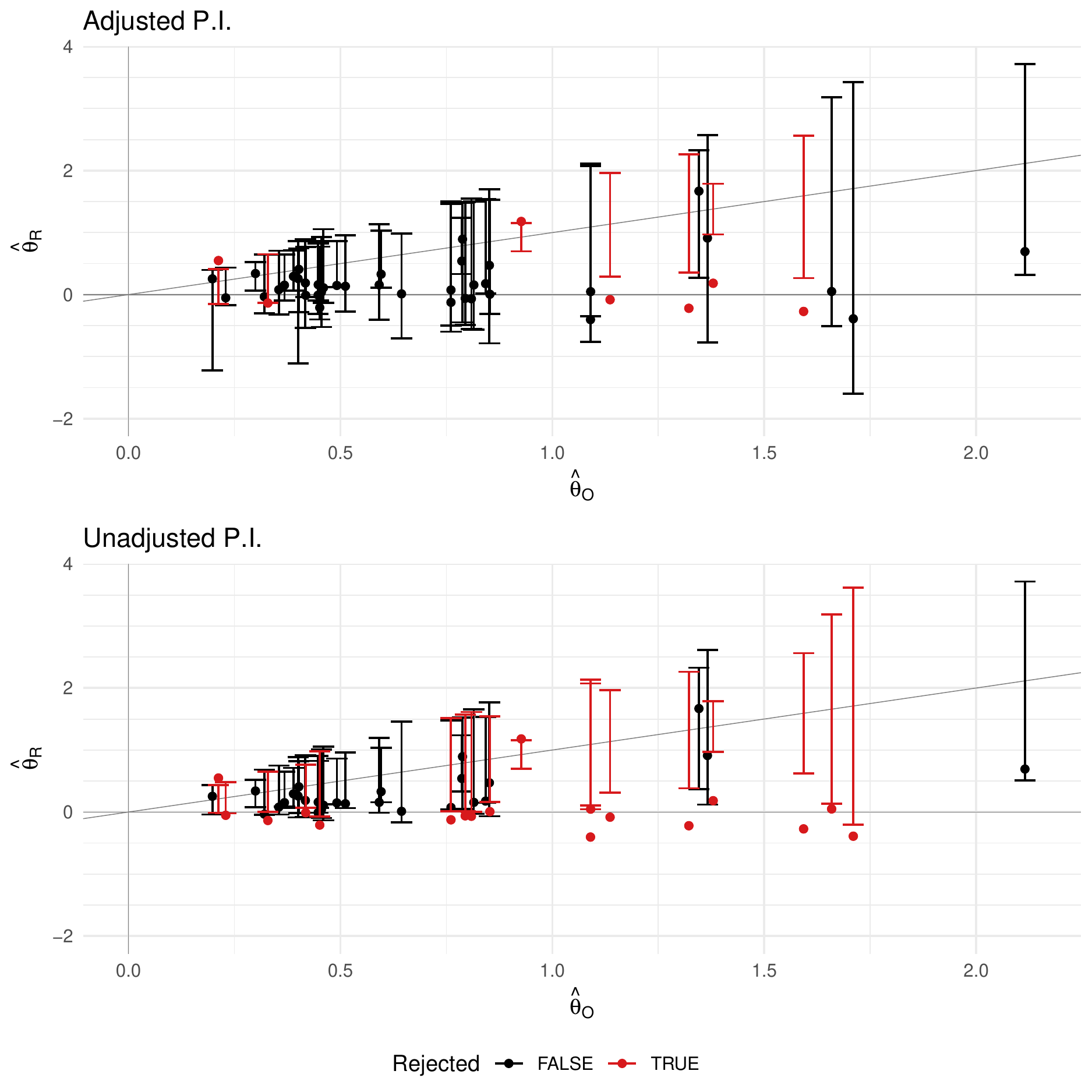}
    \caption{Predictive intervals for $\htheta_R$, both adjusted and unadjusted for selection, overlay with a plot of $\htheta_R$ against $\htheta_O$. Studies 36 and 145 are not shown here. By definition we reject $H_0: \theta_O = \theta_R$ whenever the replication effect size estimate lies outside of the predictive interval. The intervals are generally longer after adjusting for selection.}
  \label{fig:pi}
  \end{figure}
  \begin{figure}[htbp]
    \centering
    \includegraphics[width=\hsize]{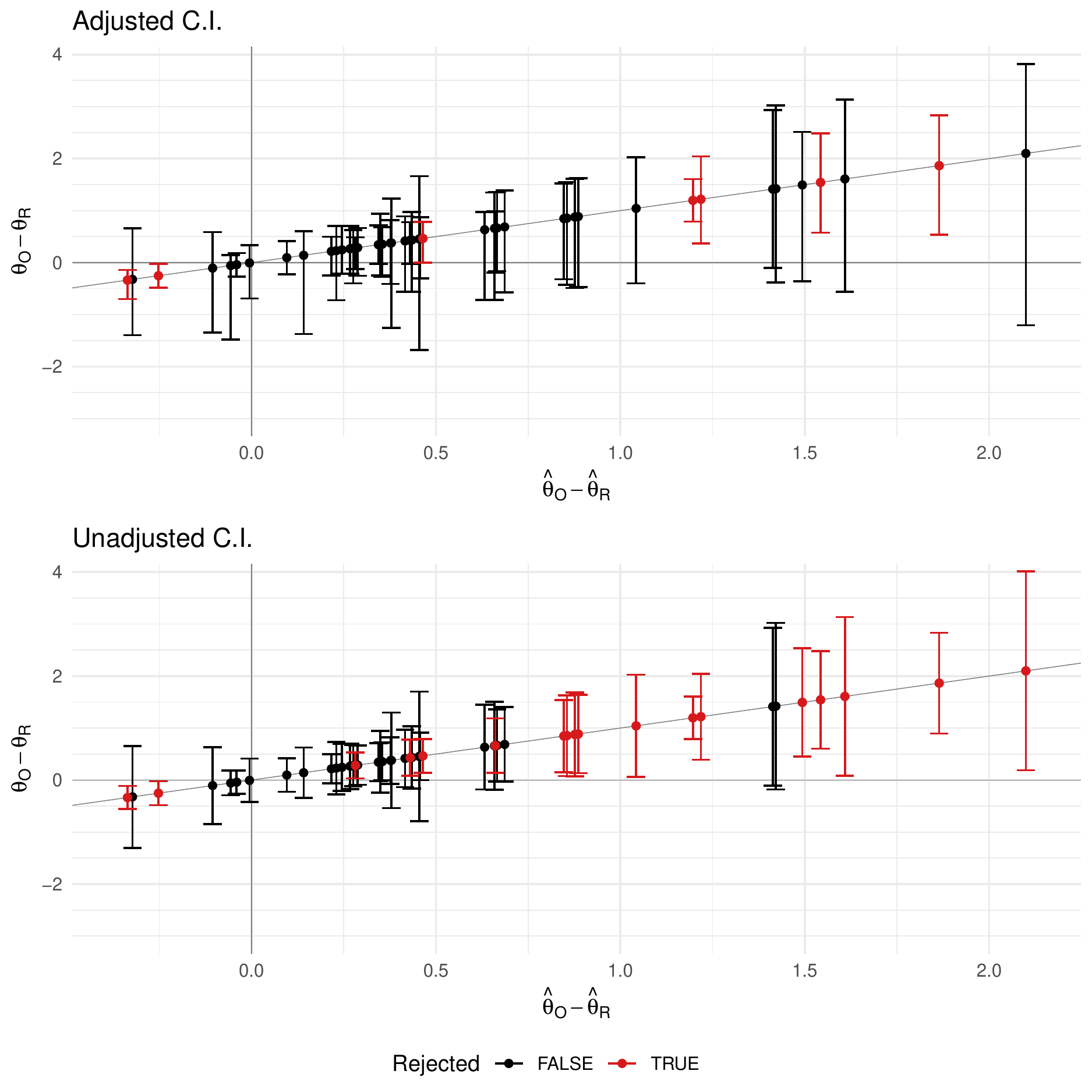}
    \caption{Confidence intervals for $\theta_O - \theta_R$, both adjusted and unadjusted for selection. By construction the null hypothesis $H_0: \theta_O = \theta_R$ is rejected when the confidence interval does not include $0$. Many of the adjusted intervals are fairly long as either the replication studies suffer low power or the original effect size estimate is near the rejection threshold. The intervals are generally longer after adjusting for selection.}
  \label{fig:ci}
  \end{figure}

  If all procedures are replicated perfectly, we should expect to reject $5\%$ of the tests on average, rather than the observed $15\%$, and after the Benjamini--Hochberg correction, there would be no rejection with $90\%$ probability. In other words, while selection bias can partly explain the discrepancies between the original and replication studies, it does not explain all of it. Nevertheless, the RP:P data cannot be taken as strong evidence of widespread failure by replication teams to satisfactorily repeat the same experiment performed in the original study. The lack of strong evidence is hardly surprising: if the original study lacks power \citep{Morey:2017} or $\htheta_O$ is closed to the rejection boundary, little can be said about $\theta_O$ and hence $\theta_O - \theta_R$. Furthermore, the replication sample sizes were determined based on the original effect size to achieve at least $80\%$ in power. Selection bias inflated the original effect size, leading to lower test power and statistically insignificant replications \citep{Etz:2016gx,Camerer:2018de}. The lack of information about $\theta_O - \theta_R$ is evident in generally wider confidence intervals after adjustment in \Cref{fig:ci}.

\subsection{Effect decline}

  Finally, we considered the proportion of effect sizes that declined. Using the selective $z$-test, we tested the hypothesis $H^D$, conditioning on the event where the $z$-scores are observed and the variable $S$. The resulting $p$-values are given in \Cref{fig:effect-decline}. Our underestimate and overestimate are $35\%$ ($=16/46$) and $100\%$ respectively, with a $90\%$ confidence interval of $(11\%, 100\%)$.
  \begin{figure}[htbp]
    \centering
    \includegraphics[width=0.8\textwidth]{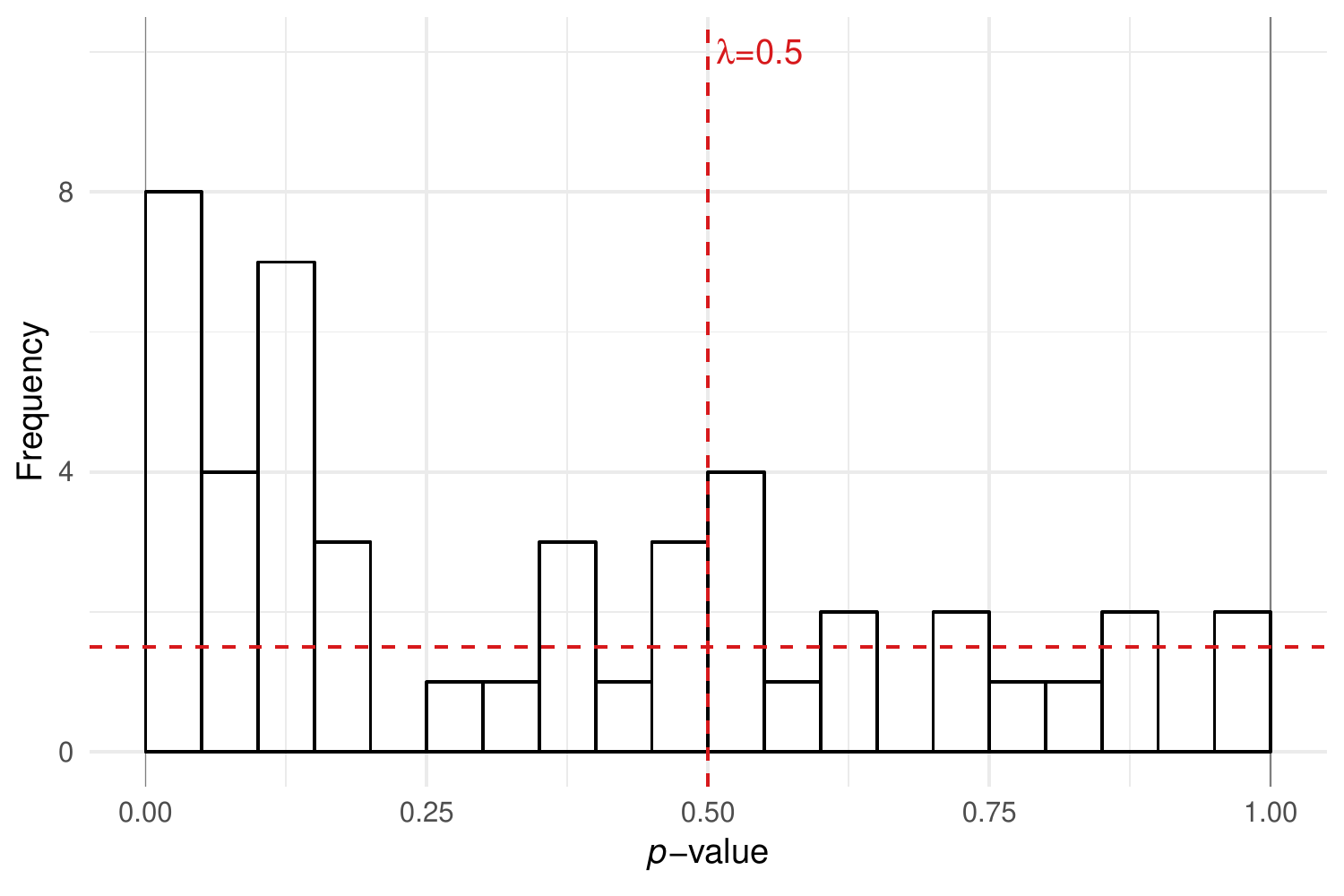}
    \caption{Histogram of the $p$-values for the null hypothesis $\theta_R \ge \theta_O$. $p$-values to the left gives more evidence for $\theta_R < \theta_O$ whereas $p$-values to the right gives more evidence for $\theta_R \ge \theta_O$. The estimate of the expected number of null $p$-values within each bin is given by the horizontal red line.}
  \label{fig:effect-decline}
  \end{figure}

  More generally, we used the hypothesis $H^{D,\rho}$ to estimate the proportion of effect sizes that declined by at least a fraction of $\rho$. The underestimate, overestimate and the $90\%$ confidence interval are given in \Cref{fig:effect-decline-range}. For example, we estimate that $10$ of the $46$ effect sizes ($22\%$) decreased by at least $25\%$, even after adjusting for selection on measurement noise. Note that this does not exclude explanations by other forms of selection, e.g.\ selecting a large effect when there is a random effect.
  \begin{figure}[htbp]
    \centering
    \includegraphics[width=0.8\textwidth]{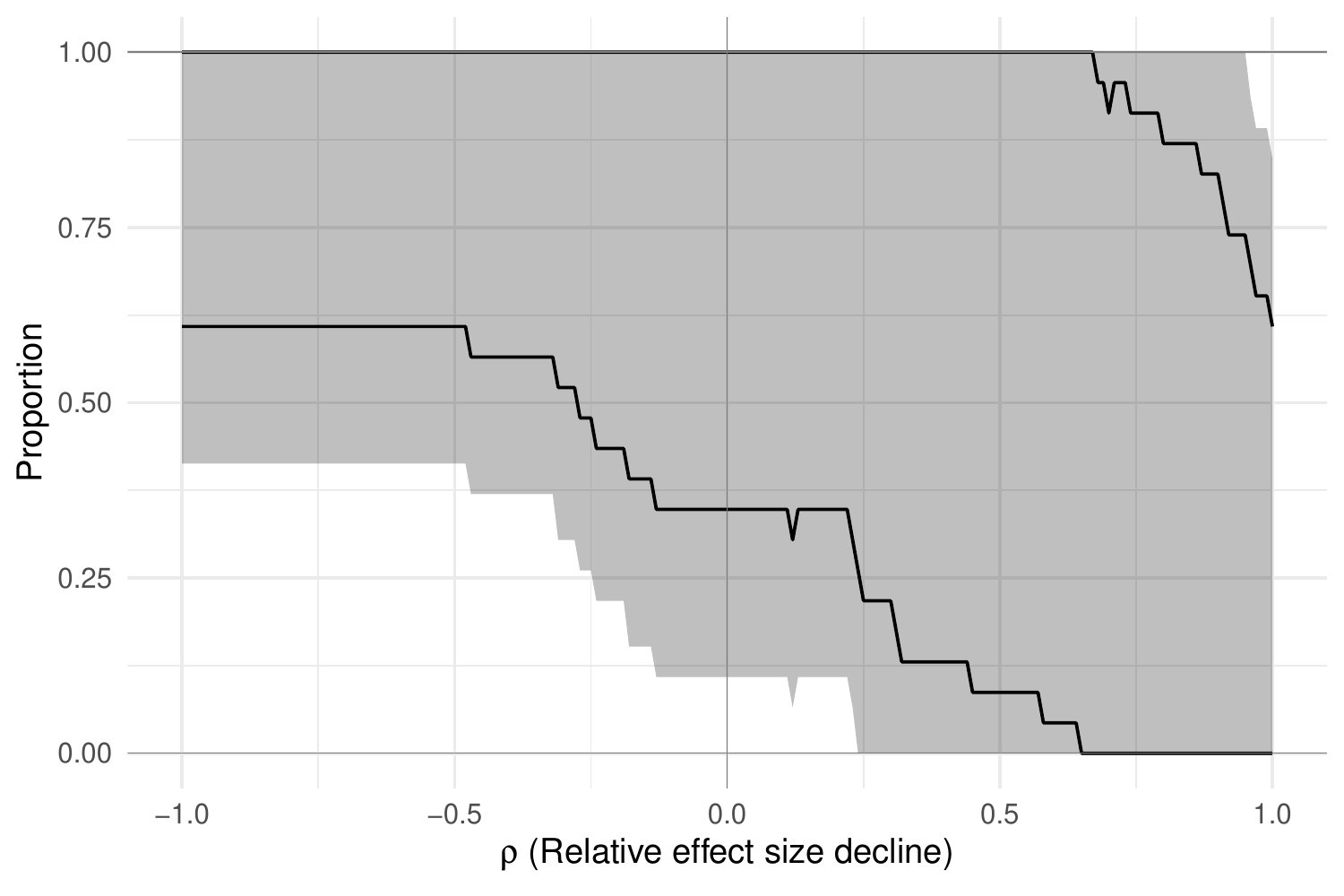}
    \caption{The underestimate, overestimate and the $90\%$ confidence interval. The lower black line is the underestimate, the high black line is the overestimate and the gray band is the $90\%$ confidence interval.}
  \label{fig:effect-decline-range}
  \end{figure}

\section{Discussion}
\label{sec:discussion}

\subsection{Importance of adjusting for selection bias}

  As we have seen, selection bias plays a powerful and pervasive role in shaping the data we observe in large-scale replication studies (and, by extension, the data we observe in published studies that have not yet been replicated!). It leads to many predictable pathologies and should be viewed as a proverbial ``elephant in the room'' whenever we discuss descriptive statistics computed from such studies. In particular, we should avoid leaping to any conclusions about how many false claims there were in the original studies, whether effect sizes declined or by how much, or which replication studies suffered from infidelities, until we have carefully ruled out the possibility that publication bias alone is to blame for whatever descriptive statistic we have computed.

  Fortunately, the truncated Gaussian model, properly combined with modern multiple testing and post-selection inference methods, opens many avenues for analyses that directly answer questions about true effect sizes with appropriate uncertainty quantification. We have explored several such avenues here \citep[see also][]{Andrews:2018vh} but many others are possible.

\subsection{Importance of statistical formality}

  In addition, we hope this article serves to advocate for the benefits of careful formal statistical modeling in analyzing replication studies, in place of (or in addition to) descriptive statistics. In particular, using vaguely specified models or eschewing models altogether can lead to analyses from which it is difficult to draw firm conclusions. For example, in \citet{OpenScienceCollaboration:2015cn}, McNemar's test was applied to a $2 \times 2$ contingency table of whether the original and replication studies are equally likely to be statistically significant. The very small $p$-value reported for this test establishes nothing more than that the original studies were selected to be statistically significant, a fact which is likely already known by most in the field. In fact, the test does not quite establish even that, because it is unclear whether this hypothesis would be true even without the effect of selection bias: The proportion of statistically significant $p$-values is a measure of the average power, which depends on the sample sizes, and the sample sizes often differed substantially between the original and replication studies.

  Another example is RP:P's use of sample correlation coefficients between independent and dependent variables as a standardized measure of effect size for comparison between the original and replication studies. This comparison implicitly assumes that the distribution of the independent variable is the same in the original and replication studies, an assumption that was violated by many of the replications. In an extreme case, an ANOVA in \citet{PurdieVaughns:2008en} with race as one of the factors used $40$ African Americans and $37$ Whites, but was replicated with $120$ African Americans and $1370$ Whites. With such a dramatic change in the distribution of an independent variable, there is no reason why the correlation coefficients should remain the same, as illustrated in the following example.

  \begin{example}
  \label{eg:unbal-t}
    A study with a two-sample $t$-test for some treatment condition is replicated. Suppose the treatment and control group are drawn from $N(1, 1)$ and $N(0, 1)$, respectively. If the ratio of the two group sizes changes from one study to another, the correlation coefficients may differ as well, even without any infidelities or hidden moderators. Borrowing the numbers from \citet{PurdieVaughns:2008en} for instance, if the original study contains $40$ treatment and $37$ control units, the true correlation coefficient is $0.45$, whereas in a replication with $120$ control and $1370$ treatment units the true coefficient is $0.26$ instead.
  \end{example}

  Replication projects similar to RP:P have since materialized, but few stated an explicit statistical hypothesis. For example, in economics, \citet{Camerer:2016cx} used the same flawed metric of proportion of statistically significant results in the original direction. A statistical analysis with explicitly stated models and hypotheses will give us more meaningful estimates, particularly valuable given how costly these large scale replication efforts are.

\subsection{Interpretation of effect shifts}

  While we have proposed several methods for quantifying discrepancies between the effect sizes in the original and replication studies, the data alone cannot tell us why they might differ. Several potential explanations include:
  \begin{enumerate}
  \item design failures, systematic biases or calculation errors in either the original or the replication study;
  \item major differences in experimental conditions between the original and replication studies, which most researchers would recognize {\em a priori} as likely to affect the results; which \citet{Gilbert:2016he} call {\em infidelities}; and
  \item minor differences in experimental conditions between the studies --- such as lighting, weather, or the passage of time --- which cannot all be controlled but whose effects may nevertheless alter the true effect size in unforeseeable ways, often referred to as {\em hidden moderators} \citep[e.g.][]{Srivastava:2015}.
  \end{enumerate}

  While there may be no sharp distinction in principle between infidelities and hidden moderators, there is a scientifically crucial difference between moderating factors that can be anticipated by experimenters and those that cannot. If we can anticipate in advance when replications are likely to fail by carefully evaluating their designs, we might hope to solve the problem simply by being more careful in setting up experiments. By contrast, if hidden moderators confound most attempts to replicate most psychological studies, it would raise profound questions about the entire enterprise of experimental psychology. In the extreme case, if even trivial changes to those conditions have large and unpredictable effects on most phenomena of interest, we might begin to despair of gaining generalizable knowledge about psychology through laboratory experimentation. 

  Our analyses point to several conclusions regarding effect shifts: First, that there are a few studies where we can be confident the effect in the replication study was significantly different than in the original study; second, that in aggregate, when effects do shift, they tend to decline (shift toward zero) in replications rather than increase; and third, that there is insufficient evidence to conclude that the vast majority of experimental effects simply evaporated upon replication. In particular, 83\% should not be treated as a reasonable estimator of the fraction of {\em true} effect sizes that declined; rather, it likely reflects that the estimates in the original studies overestimated their corresponding true effects due to selection bias.

  One possible explanation for systematically declining effect involves a subtler form of selection bias, where every experiment's effect size is random, buffeted by hidden moderators, and those experiments whose moderators primarily magnify the effect size are more likely to be published. That is, in the same way that experimenters select studies whose sampling error is large, they also selects for studies whose true effect size is larger than usual. Further systematic replication studies may help to shed light on which factors are most often the culprits in moderating true effect sizes, possibly improving the reliability of experiments and leading to new scientific insights \citep{Barrett:2015vl,Klein:2018}.

\subsection{Future work}

  As large-scale replicability studies are becoming more common in assessing the ``well-being'' of a scientific domain, this paper serves as a stepping stone for improving methodologies in future replicability studies.

  First, selection for significance is an inevitable consequence of the current scientific process. Our adjustments for selection allows not only better analysis, but also more informed design of future replication studies, e.g.\ better power calculations for and sizing of replications. While these adjustments are admittedly crude, they are necessitated by the limitations in the given data. With more available information, a better model for selection can be used. For example, with the advancement of preregistration, we can use the external comparison method to produce less conservative estimates of the directional FDP at level $\alpha = 0.05$ if we have more information about statistically nonsignificant studies. With more replications carried out, we can estimate the publication bias model in \citet{Andrews:2018vh} more precisely, which allows different propensity for publication for different statistical significant $p$-values as opposed to \Cref{ass:sig-enough}. Together with higher powered design in replications \citep[e.g.][]{Camerer:2018de}, we can enhance the precision of our estimators and power of our tests.

  Second, we emphasized the importance of statistical formality. Our proposed criteria are based on clearly defined parameters. While these criteria may not suit all needs in future replicability studies, additional formal hypotheses can also be analyzed under the post-selection inference framework similarly.

  With our proposed criteria and procedures, researchers can perform more informative inferences than the current practice, and provide a clearer picture of the replicability crisis.

\section*{Reproducibility}

  A git repository containing with the code generating the images in this article is available at \url{https://github.com/kenhungkk/assessing-replicability.git}.

\section*{Acknowledgment}

  We thank Marcel A L M van Assen, Yoav Benjamini, Dean Eckles, Philip B Stark, Jacob Steinhardt, Jonathan Taylor, Alexa Tulett, Stefan Wager, Daniel Yekutieli, and Bin Yu for helpful comments and discussions. In addition, we are grateful to our two anonymous reviewers for their helpful suggestions which have improved the paper.

\begin{supplement}
  \sname{Supplement A}\label{suppA}
  \stitle{Supplement to ``Statistical Methods for Replicability Assessment''}
  \slink[url]{https://github.com/kenhungkk/assessing-replicability/raw/public/supplement.pdf}
  \sdescription{We evaluate our approximation of $t$-distributions by normal distributions, as well as detail considerations made for individual studies.}
\end{supplement}

\Urlmuskip=0mu plus 1mu\relax
\bibliographystyle{imsart-nameyear}
\bibliography{papers,additional}

\end{document}